\documentclass[journal,twoside,web]{ieeecolor}
\usepackage{generic}
\def\BibTeX{{\rm B\kern-.05em{\sc i\kern-.025em b}\kern-.08em
    T\kern-.1667em\lower.7ex\hbox{E}\kern-.125emX}}
\usepackage{cite}
\usepackage{amsmath,amssymb,amsfonts}
\usepackage{algorithmic}
\usepackage{algorithm}
\usepackage{graphicx}
\usepackage{eso-pic}
\usepackage{hyperref}
\hypersetup{hidelinks=true}
\usepackage{textcomp}
\usepackage{xcolor}
\usepackage{subcaption}   
\usepackage{pgfplots}
\pgfplotsset{compat=1.18}

\newtheorem{assumption}{Assumption}
\newtheorem{definition}{Definition}
\newtheorem{lemma}{Lemma}
\newtheorem{theorem}{Theorem}
\newtheorem{remark}{Remark}
\newtheorem{proposition}{Proposition}
\newtheorem{corollary}{Corollary}
\newtheorem{problem}{Problem}

\definecolor{tacBlue}{HTML}{006BA4}
\definecolor{tacOrange}{HTML}{FF800E}
\definecolor{tacDarkGray}{HTML}{595959}
\definecolor{tacLightGray}{HTML}{666666}

\pgfplotsset{compat=1.18}

\pgfplotsset{
  tacaxis/.style={
    width=0.93\linewidth,
    height=0.6\linewidth,
    xlabel={noise scaling $s$},
    ylabel={spectral radius $\rho$},
    xmin=0, xmax=6,
    tick align=outside,
    tick style={black},
    axis line style={black},
    line width=0.8pt,
    grid=major,
    major grid style={draw=gray!20, line width=0.2pt},
    legend cell align=left,
    legend style={
      draw=none,
      fill=none,
      font=\footnotesize,
      at={(0.02,0.98)},
      anchor=north west,
      row sep=1pt,
      inner sep=1pt
    },
    label style={font=\footnotesize},
    tick label style={font=\footnotesize},
  },
seqAstyle/.style={
    thick,
    tacBlue,
    dashdotted,
  },
  seqBstyle/.style={
    thick,
    tacOrange,
    dashed,
  },
  decstyle/.style={
    thick,
    tacDarkGray,
    solid,
  },
  thresholdstyle/.style={
    semithick,
    tacLightGray,
    dotted,
  },
  seqAdensestyle/.style={
    seqAstyle,
    dash pattern={on 2pt off 1pt on 0.8pt off 1pt},
  },
  seqBdensestyle/.style={
    seqBstyle,
    dash pattern={on 1.6pt off 0.9pt},
  },
  predstyle/.style={
    black,
    line width=0.7pt,
    densely dotted,
  }
}

\begin{document}

\title{Policy Iteration for Linear-Quadratic Stochastic Differential Games with State- and Control-Dependent Noise}

\author{Karl Handwerker, Felix Th\"ommes, \IEEEmembership{Member, IEEE,} Lucas G\"unther, Balint Varga, \IEEEmembership{Member, IEEE,} and S\"oren Hohmann, \IEEEmembership{Senior Member, IEEE}
\thanks{All authors are with the Institute of Control Systems (IRS),
Karlsruhe Institute of Technology (KIT), 76131 Karlsruhe, Germany.
Corresponding author: K.~Handwerker (karl.handwerker@kit.edu).}}

\maketitle

\AddToShipoutPictureBG*{%
  \AtPageLowerLeft{\raisebox{0.32in}{\makebox[\paperwidth][c]{%
    \parbox{\textwidth}{\centering\footnotesize
      This work has been submitted to the IEEE for possible publication.
      Copyright may be transferred\\
      without notice, after which this version may no longer be accessible.}}}}%
}

\begin{abstract}
This paper presents a novel sequential policy iteration (PI) method for stochastic differential games with state- and control-dependent noise. 
The updates preserve mean-square stability, so that the iteration is well posed. 
We further derive a closed-form expression for the Fréchet derivative of the sequential PI map at a Nash equilibrium.
The resulting characterization reveals how control-dependent noise, policy-evaluation sensitivity, and update ordering govern local error propagation, and yields explicit sufficient conditions for local linear convergence.
Since finding an initial stabilizing solution is a major challenge in policy iteration, we also propose a homotopy-based initialization that ensures a valid starting point.
The effectiveness of the proposed PI algorithm and the analytical results are verified through a numerical example.
\end{abstract}

\begin{IEEEkeywords}
Differential games, Mean-square stability, Multiplicative noise, Nash equilibrium, Policy iteration, Riccati equations 
\end{IEEEkeywords}

\section{Introduction}
\label{sec:introduction}

Noncooperative stochastic differential games are mathematical models of the dynamic interaction between multiple decision makers acting on an uncertain system, each pursuing an individual objective \cite{basar_olsder_1998,friedman_2013,hu_lauriere_2024_survey_ml_stochastic_games}.
Originating in deterministic differential game theory \cite{friedman_2013}, the stochastic formulation captures uncertainties inherent in a wide range of applications such as financial markets \cite{saito_takahashi_2019_automatica}, networked control \cite{liu_huang_2025_tac_aggregative,yuksel_basar_2013}, and human--machine interaction \cite{varga_inga_lemmer_hohmann_2021_ccta,kille_leibold_karg_varga_hohmann_2024_roman}, where the uncertainty typically scales with the state and with the applied controls \cite{todorov_2005_neural_computation}.

A central solution concept in such settings is the feedback Nash equilibrium, namely a profile of state-feedback strategies for which no player can improve its own objective by a unilateral deviation \cite{starr_ho_1969,basar_olsder_1998}.
The computation of a feedback Nash equilibrium requires solving a system of coupled equilibrium conditions, which in general take the form of coupled stochastic Hamilton--Jacobi--Bellman equations \cite{yong_zhou_1999,basar_olsder_1998}.
In infinite-horizon linear--quadratic (LQ) settings, these conditions reduce to a system of coupled stochastic algebraic Riccati equations that characterize the equilibrium value functions and the associated linear feedback gains \cite{zhu_zhang2013_state_control_noise}.

Closed-form solutions to these equations are available only in special cases, so that feedback Nash equilibria must in general be computed numerically \cite{basar_olsder_1998}.
While the single-player counterpart can be solved by classical iterative methods \cite{kleinman_1969,mclane_1971,damm2004}, the coupled equilibrium conditions of a game constitute a fixed-point problem rather than an optimization problem, so that the underlying monotonicity arguments no longer apply.
Assessing convergence remains an open question for the computation of Nash equilibria in general \cite{hu_lauriere_2024_survey_ml_stochastic_games}, since nonzero-sum games may possess several stabilizing feedback Nash equilibria already in the deterministic special case \cite{basar_olsder_1998,engwerda2005}.
Global convergence to a distinguished equilibrium can therefore not be expected without strong additional assumptions, and convergence becomes an inherently local and equilibrium-dependent question.
Even at a given equilibrium, iterations of the same method class converge at different rates, since the multiplayer structure allows the players to be updated either simultaneously or sequentially one at a time in a chosen order \cite{nortmann_monti_sassano_mylvaganam2024_tac_iterative_datadriven_lq_games}, the latter being empirically reported to converge faster \cite{chen_chen_lewis_xie_mcpi_nash_dg,chen_chen_lewis_2026_automatica_mcpi_nonlinear}. 
A structural explanation of this phenomenon in policy iteration is missing.

The solution methods themselves, by contrast, have advanced considerably, though largely for the deterministic special case, in which iterative schemes built around repeated Lyapunov or Riccati equation solutions are well established \cite{freiling_1996,li_gajic_1995}.
Policy iteration (PI), which alternates between policy evaluation and policy improvement, is the most widely used approach for computing feedback Nash equilibria in deterministic games \cite{puterman_brumelle_1979_pi_newton,lewis_vrabie_2009_rl_adp, nortmann_monti_sassano_mylvaganam2024_tac_iterative_datadriven_lq_games,guan_salizzoni_kamgarpour_summers_2024_pi_lqdg,chen_chen_lewis_xie_mcpi_nash_dg,vamvoudakis_lewis_2011_online_hjb_games,thoemmes2026_ifac_policy_gradient}. 
For stochastic games, iterative methods are available for two-player formulations \cite{ivanov_tanov_2018_iterative_lq_stochastic_games,sun_jiang_zhang2012_discrete_lq_games}, for state-dependent noise \cite{sagara_mukaidani_yamamoto2007_state_dep_noise,mukaidani2009_automatica_soft_constrained}, and for multiplayer games over finite time horizons \cite{han_hu_2020_deep_fp,han_hu_long_2020_convergence_dfp,hambly_xu_yang2022_policy_gradient_lq_games,andersson_andersson_ljung2026_fp_fbsde}.

None of these settings, however, combines an arbitrary number of players with noise that depends on the controls over an infinite horizon.
Control-dependent noise in particular couples the players' stationarity conditions and turns the equilibrium conditions into rational matrix equations, whereas they are polynomial in the deterministic case and under purely state-dependent noise \cite{damm2004,zhu_zhang2013_state_control_noise}.
Whenever the diffusion does not vanish, drift and diffusion jointly govern the evolution of the state variance, so that mean-square stability (MSS) replaces asymptotic stability as the relevant closed-loop property \cite{oksendal_2003}.
As a consequence, neither the algorithms developed for the deterministic and state-dependent case nor the results establishing their convergence carry over, since they address different equations and, in the deterministic case, certify a different notion of stability.
In fact, a recent survey identifies the general solvability of Riccati-type equations for the explicit implementation of feedback-type optimal solutions in games with state- and control-dependent noise as an open research direction \cite{moon_wang_basar_2026_survey_lq_sdg}.

We therefore develop a policy-iteration method for infinite-horizon \(N\)-player nonzero-sum stochastic LQ differential games with state- and control-dependent noise.
The policy evaluation step presupposes a mean-square-stabilizing gain profile, since otherwise the infinite-horizon value equations and costs are not defined \cite{willems_willems_1976}.
How such a profile is to be obtained has been raised as a question in its own right \cite{nortmann_monti_sassano_mylvaganam2024_tac_iterative_datadriven_lq_games}.
Homotopy continuation has been used for this purpose in deterministic single-player \cite{chen_lewis_li_2022_homotopic_pi} and game \cite{chen_chen_lewis_xie_mcpi_nash_dg} settings.
Both rely on a different notion of stability and require each player to stabilize the system on its own, so that neither extends to the present class.
Once started, the iteration must then retain mean-square stability at every update, which renders every iterate implementable as a stabilizing feedback.

To the best of our knowledge, the sequential policy-iteration method proposed here is the first for this class of games. 
Our main contributions are:
\begin{enumerate}
    \item A policy-iteration method whose player-wise updates
    preserve mean-square stability, with fixed points that are precisely
    the stabilizing feedback Nash equilibria.

    \item A local convergence analysis that delimits the equilibria at which sequential updating is guaranteed to contract at least as strongly as simultaneous updating, and shows that with three or more players the update order can decide whether an equilibrium is locally attracting at all.

    \item A homotopy-based initialization that produces a mean-square-stabilizing
    initial gain profile in finitely many continuation steps whenever the game is mean-square stabilizable.
\end{enumerate}

\noindent \textbf{Notation:}
For positive integers \(n\) and \(m\), let \(\mathbb R^n\) and \(\mathbb R^{n\times m}\) denote the space of real vectors of dimension \(n\) and the space of real \(n\times m\) matrices, respectively. Let \(\mathbb S^n\) denote the set of real symmetric \(n\times n\) matrices, and define
$
\mathbb S^n_+ := \{X\in\mathbb S^n : X\succeq 0\},
$ and
$\mathbb S^n_{++} := \{X\in\mathbb S^n : X\succ 0\}.
$
For \(X\in\mathbb S^n\), \(\lambda_{\min}(X)\) and \(\lambda_{\max}(X)\) denote its smallest and largest eigenvalues, respectively. For a square matrix \(X\), \(\rho(X)\) denotes its spectral radius. We use \(\|\cdot\|\) for the Euclidean norm on vectors and the associated induced operator norm on matrices. The identity matrix of dimension \(n\) is denoted by \(I_n\), or simply by \(I\) when the dimension is clear from the context.
Let $(\Omega,\mathcal F,\mathbb P)$ be a probability space, and let $\{\mathcal F_t\}_{t\ge0}$ be a right-continuous filtration such that $\mathcal F_0$ contains all $\mathbb P$-null sets of $\mathcal F$. Assume that $W=(W_t)_{t\ge0}$ is a one-dimensional standard Brownian motion adapted to $\{\mathcal F_t\}_{t\ge0}$. Throughout the paper, the initial state is deterministic, $x(0)=x_0\in\mathbb R^n$.
Whenever an equilibrium notion is stated, it is understood for every initial state $x_0\in\mathbb R^n$.

\section{Problem Formulation}\label{sec:problem_statement}

We consider an $N$-player stochastic differential game.
For each player $i\in \mathcal{I} := \{1,\dots,N\}$, we define the control $u_i(t) \in \mathbb{R}^{m_i}$, where $u(t):=(u_1(t),\dots,u_N(t))$ denotes the joint control. The state $x(t)\in\mathbb R^n$ evolves
according to the stochastic differential equation (SDE) 
\begin{align}\label{eq:sde-N}
dx(t)
&=
\Bigl(Ax(t)+\sum_{i\in\mathcal I} B_i u_i(t)\Bigr)\,dt\nonumber\\
&\quad+
\Bigl(Cx(t)+\sum_{i\in\mathcal I} D_i u_i(t)\Bigr)\,dW(t),
\end{align}
where $A,C\in\mathbb{R}^{n\times n}$ and
$B_i,D_i\in\mathbb{R}^{n\times m_i}$, $i \in \mathcal I$ are constant matrices.
Each player $i \in \mathcal I$ seeks to minimize the infinite-horizon quadratic cost
\begin{equation}
\label{eq:cost}
J_i(u;x_0)
=
\mathbb E\int_0^\infty
x(t)^\top Q_i x(t)
+
\sum_{j\in\mathcal I} u_j(t)^\top R_{ij}u_j(t)
\,dt.
\end{equation}
where $Q_i\in\mathbb{S}^n_{++}$, $R_{ii}\in\mathbb S^{m_i}_{++}$, and $R_{ij}\in\mathbb S^{m_j}_{+}$ for all $j\in\mathcal{I}\setminus\{i\}$.

We consider constant linear state-feedback strategies of the form
$
u_i(t)=-K_i x(t),
$
with $
K_i\in\mathcal K_i:=\mathbb R^{m_i\times n}$,
$ i\in\mathcal I$,
and define the gain profile space
$
\mathcal K:=\prod_{i\in\mathcal I}\mathcal K_i,
$
whose elements are the $N$-tuples $K=(K_i)_{i\in\mathcal I}$. For $i\in\mathcal I$, we also write $K_{-i}:=(K_j)_{j\in\mathcal I\setminus\{i\}}$ for the gains of all players except player $i$.
For $K\in\mathcal K$, define the closed-loop drift and diffusion matrices
\begin{equation}
A_{\mathrm{cl}}(K)
:=
A-\sum_{i\in\mathcal I} B_iK_i,
\quad
C_{\mathrm{cl}}(K)
:=
C-\sum_{i\in\mathcal I} D_iK_i,
\end{equation}
such that the closed-loop dynamics becomes
\begin{equation}
\label{eq:closed-loop-SDE}
dx(t)
=
A_{\mathrm{cl}}(K)x(t)\,dt
+
C_{\mathrm{cl}}(K)x(t)\,dW(t), \, x(0)=x_0.
\end{equation}
For every gain profile \(K\in\mathcal K\), the closed-loop system \eqref{eq:closed-loop-SDE} admits a unique strong solution on \([0,\infty)\).
The induced closed-loop control is $u^K(t):=(-K_i x(t))_{i\in\mathcal I}$, and we define
\begin{equation}
J_i(K;x_0):=J_i(u^K;x_0)
\end{equation}
for the induced cost of player $i \in \mathcal I$ with the convention that \(J_i(K;x_0)=+\infty\) whenever the integral in \eqref{eq:cost} diverges under \(u^K\). The associated value function under $K$ is
defined by
\begin{equation}\label{eq:value-function}
V_i^K(x_0):=J_i(K;x_0).
\end{equation}

\begin{definition}[Mean-square stability {\cite[Def.~1.5.1]{damm2004}}]
The closed-loop system \eqref{eq:closed-loop-SDE} is called mean-square stable (MSS) if its state trajectory satisfies
\begin{equation}
\lim_{t\to\infty}\mathbb{E}\|x(t)\|^2=0
\quad \forall x_0\in\mathbb{R}^n.
\end{equation}
The system \eqref{eq:sde-N} is called mean-square stabilizable if there exists an $N$-tuple \(K\in\mathcal K\) such that the closed-loop system \eqref{eq:closed-loop-SDE} is MSS.
\end{definition}

We denote the set of stabilizing gain profiles by
\begin{equation}\label{eq:FN}
\mathcal K_{\mathrm{MSS}}
:=
\left\{
K\in\mathcal K:
\eqref{eq:closed-loop-SDE}
\text{ is MSS}
\right\}.
\end{equation}

\begin{assumption}[Mean-square stabilizability]\label{ass:MSS}
The stochastic game \eqref{eq:sde-N} is mean-square stabilizable by constant
linear state feedback, that is, $\mathcal K_{\mathrm{MSS}}\neq\varnothing$.
\end{assumption}

Following \cite{damm2004}, a gain profile
$K\in\mathcal K$ belongs to $\mathcal K_{\mathrm{MSS}}$ if and only if there exists a matrix
$X=X^\top\succ0$ such that
\begin{equation}\label{eq:Lyap-MSS}
A_{\mathrm{cl}}(K)^\top X
+
X A_{\mathrm{cl}}(K)
+
C_{\mathrm{cl}}(K)^\top X C_{\mathrm{cl}}(K)
\prec0.
\end{equation}
In this setting, the second moment \(\mathbb{E}\|x(t)\|^2\) decays moreover exponentially; see, e.g.,
\cite[Thm.~1.5.3]{damm2004}.

\begin{definition}[Stabilizing linear feedback Nash equilibrium]\label{def:LFNE}
A gain profile $K^\star=(K_i^\star)_{i\in\mathcal I}\in\mathcal K_{\mathrm{MSS}}$ is called a
stabilizing linear feedback Nash equilibrium \cite{basar_olsder_1998,zhu_zhang2013_state_control_noise} if, for every player
$i\in\mathcal I$ and every other gain $\widetilde K_i\in\mathcal K_i$,
\begin{equation}\label{eq:nash-def}
J_i(K^\star;x_0)
\le
J_i\bigl((\widetilde K_i,K_{-i}^\star);x_0\bigr)
\qquad
\forall x_0\in\mathbb R^n,
\end{equation}
whenever $(\widetilde K_i,K_{-i}^\star)\in\mathcal K_{\mathrm{MSS}}$.
\end{definition}

If \(K^\star\in\mathcal K_{\mathrm{MSS}}\) is a stabilizing linear feedback Nash equilibrium,
then for each player \(i\in\mathcal I\) there exists a matrix
$
P_i^\star\in\mathbb S_{++}^n
$
such that
$
V_i^{K^\star}(x_0)=x_0^\top P_i^\star x_0$,
 $\forall x_0\in\mathbb R^n,
$
and the matrices \(P_i^\star\) satisfy, for each player \(i\in\mathcal I\), the coupled relations
\begin{align}
0
&=
A_{\mathrm{cl}}(K^\star)^\top P_i^\star
+
P_i^\star A_{\mathrm{cl}}(K^\star)
+
C_{\mathrm{cl}}(K^\star)^\top P_i^\star C_{\mathrm{cl}}(K^\star)
\nonumber\\
&\quad
+
Q_i
+
\sum_{j\in\mathcal I}(K_j^\star)^\top R_{ij}K_j^\star,
\label{eq:equilibrium-eval}
\end{align}
and
\begin{equation}\label{eq:equilibrium-improve}
\bigl(R_{ii}+D_i^\top P_i^\star D_i\bigr)K_i^\star
=
B_i^\top P_i^\star
+
D_i^\top P_i^\star
\Bigl(
C-\sum_{j\in\mathcal I\setminus\{i\}} D_jK_j^\star
\Bigr).
\end{equation}
These coupled equilibrium equations \eqref{eq:equilibrium-eval}--\eqref{eq:equilibrium-improve} follow from the quadratic structure of the closed-loop value functions \eqref{eq:value-function} and the first-order optimality conditions of the induced best-response problems; see, e.g., \cite{zhu_zhang2013_state_control_noise}.
They admit closed-form solutions only in special cases, so that \(K^\star\) must in general be computed numerically \cite{basar_olsder_1998}.

\begin{problem}\label{prob:main}
Develop a policy iteration algorithm for computing a stabilizing linear feedback Nash equilibrium $K^\star \in \mathcal{K}_{\mathrm{MSS}}$ for the stochastic $N$-player differential game \eqref{eq:sde-N}--\eqref{eq:cost}, that is, solve the coupled equilibrium system \eqref{eq:equilibrium-eval}--\eqref{eq:equilibrium-improve} via a well-posed PI scheme on the set of stabilizing gain profiles $\mathcal{K}_{\mathrm{MSS}}$.
\end{problem}

\begin{problem}\label{prob:convergence}
Characterize the local convergence behavior of the proposed algorithm near a stabilizing linear feedback Nash equilibrium $K^\star \in \mathcal{K}_{\mathrm{MSS}}$, and provide verifiable sufficient conditions for local convergence.
\end{problem}

\begin{problem}\label{prob:init}
Construct a gain profile $K\in\mathcal K_{\mathrm{MSS}}$ directly from the data of the game \eqref{eq:sde-N}--\eqref{eq:cost}, so that the proposed algorithm can be initialized without a stabilizing profile being known in advance.
\end{problem}

\section{Sequential Policy Iteration}\label{sec:03_CSPI}

This section addresses Problem~\ref{prob:main} by presenting a policy-iteration method based on sequential player-wise updates for computing stabilizing linear feedback Nash equilibria of the stochastic $N$-player differential game \eqref{eq:sde-N}--\eqref{eq:cost}.
Starting from a stabilizing gain profile, the method alternates between (A)~policy evaluation via a generalized Lyapunov equation and (B)~policy improvement based on fixed-value stationarity conditions.
Within each iteration, the players are updated one at a time in a fixed order. Each player's improvement step is evaluated at a partially updated gain profile that already incorporates the most recent updates of the preceding players, so that the following players within the same iteration already benefit from the updated gains of preceding players.
Both the policy-evaluation step and the associated quadratic value representation are defined only for gain profiles in \(\mathcal K_{\mathrm{MSS}}\). The sequential construction is therefore designed to preserve mean-square stability at each intermediate player-wise update, so that the policy-iteration mapping remains well defined on this set.

\subsection{Policy Evaluation}\label{sec:eval_model_based}

The policy evaluation step assesses a given stabilizing and in general suboptimal gain profile $K\in\mathcal K_{\mathrm{MSS}}$ by computing, for each player $i \in \mathcal I$, the infinite-horizon cost \eqref{eq:cost} induced by the closed-loop system \eqref{eq:closed-loop-SDE}.
Under $u_i(t)=-K_i x(t)$, define the induced running-cost matrix
\begin{equation}
S_i(K):=Q_i+\sum_{j\in\mathcal I} K_j^\top R_{ij}K_j,
\end{equation}
so that
\begin{equation}\label{eq:induced_running_cost}
J_i(K;x_0)=\mathbb{E}\!\int_0^\infty x(t)^\top S_i(K)x(t)\,dt.
\end{equation}

\begin{lemma}[Policy evaluation under fixed linear feedback]
\label{lem:eval-Pi}
Let \(K\in\mathcal K_{\mathrm{MSS}}\) and fix \(i\in\mathcal I\). Then the generalized Lyapunov equation
\begin{equation}
\label{eq:eval-Ei}
A_{\mathrm{cl}}(K)^\top P_i
+ P_i A_{\mathrm{cl}}(K)
+ C_{\mathrm{cl}}(K)^\top P_i C_{\mathrm{cl}}(K)
+ S_i(K)=0
\end{equation}
admits a unique symmetric solution \(P_i(K)\in\mathbb S^n\). Moreover,
\(P_i(K)\in\mathbb S^n_{++}\), and the induced cost of player \(i\) admits the quadratic value
representation
$
J_i(K;x_0)=x_0^\top P_i(K)x_0, 
\, \forall x_0\in\mathbb R^n.
$

\end{lemma}

\begin{proof}
Fix \(K\in\mathcal K_{\mathrm{MSS}}\) and \(i\in\mathcal I\). Existence and
uniqueness of a symmetric solution \(P_i(K)\in\mathbb S^n\) of
\eqref{eq:eval-Ei} follow from the unique solvability of the generalized
Lyapunov equation for mean-square stable closed-loop systems; see, e.g.,
\cite{damm2004}. Since \(Q_i\succ0\) and \(R_{ij}\succeq0\), we have
\(S_i(K)\succ0\), and therefore \(P_i(K)\succ0\).

Applying It\^o's formula to \(x(t)^\top P_i(K)x(t)\) and using
\eqref{eq:eval-Ei} gives \(d\bigl(x(t)^\top P_i(K)x(t)\bigr)
=-x(t)^\top S_i(K)x(t)\,dt+dM_i(t)\), where \(M_i(\cdot)\) is a local
martingale; see, e.g., \cite{oksendal_2003}. Fix \(T>0\) and let
\(\tau_\ell:=\min\bigl(\inf\{t\ge0:\|x(t)\|\ge\ell\},\,T\bigr)\) for
\(\ell>0\). Integrating over \([0,\tau_\ell]\) and taking expectations, using
that \(x(0)=x_0\) is deterministic and that the stopped process
\(M_i(\min(\cdot,\tau_\ell))\) is a true martingale, we obtain
\begin{align*}
&\mathbb E\bigl[x(\tau_\ell)^\top P_i(K)x(\tau_\ell)\bigr]
- x_0^\top P_i(K)x_0 \\
&\qquad = -\,\mathbb E\!\int_0^{\tau_\ell} x(t)^\top S_i(K)x(t)\,dt.
\end{align*}
Since \(\tau_\ell\le T\) and \(\tau_\ell\to T\) almost surely as
\(\ell\to\infty\), monotone convergence applies to the right-hand side, and
since \(x(\cdot)\) is continuous with
\(\mathbb E\bigl[\sup_{t\in[0,T]}\|x(t)\|^2\bigr]<\infty\) for the linear SDE
\eqref{eq:closed-loop-SDE}, see, e.g., \cite{nisio2015}, dominated convergence
applies to the left-hand side. Hence the identity holds with \(\tau_\ell\)
replaced by \(T\). Since \(K\in\mathcal K_{\mathrm{MSS}}\), one has
\(\mathbb E\|x(T)\|^2\to0\), and letting \(T\to\infty\) yields
\(J_i(K;x_0)=x_0^\top P_i(K)x_0\).
\end{proof}

\subsection{Player-Wise Policy Improvement}
We first develop the player-wise policy-improvement step. Given a value matrix, it returns the gain that enforces the corresponding fixed-value stationarity condition, independently of how the value matrix was obtained.

Fix a player \(i\in\mathcal I\), a value matrix \(P\in\mathbb S^n_{++}\), and the gains \(K_j\), \(j\in\mathcal I\setminus\{i\}\), of the remaining players. The improved gain for player \(i\) is obtained by enforcing the first-order stationarity condition of player \(i\)'s induced cost, which is quadratic in \(K_i\), while holding the value matrix fixed at \(P\). This fixed-value stationarity condition is
\begin{equation}\label{eq:Ui-stationarity}
\begin{aligned}
0
&=
\bigl(R_{ii}+D_i^\top P D_i\bigr)\,\widehat K_i \\
&\,
+\sum_{j\in\mathcal I\setminus\{i\}}
D_i^\top P D_j\,K_j
- B_i^\top P
- D_i^\top P C,
\end{aligned}
\end{equation}
where \(\widehat K_i\in\mathbb R^{m_i\times n}\) denotes the updated gain of player \(i\). Since \(R_{ii}\succ0\) and \(P\succ0\), the matrix
$
R_{ii}+D_i^\top P D_i
$
is positive definite. Hence \eqref{eq:Ui-stationarity} uniquely determines \(\widehat K_i\) as a function of the value matrix \(P\) and the gains \(K_j\), \(j\in\mathcal I\setminus\{i\}\).

\begin{algorithm}[t]
\caption{Sequential Policy Iteration for Stochastic $N$-Player LQ Differential Games}
\label{alg:cascaded-PI}
\begin{algorithmic}
\REQUIRE Stabilizing initial gain profile $K^0=(K_1^0,\dots,K_N^0)\in\mathcal K_{\mathrm{MSS}}$,
system matrices $A,B_i,C,D_i$, weights $Q_i,R_{ii},R_{ij}$, tolerance $\varepsilon>0$

\STATE $k\gets 0$

\REPEAT
    \FOR{$i \in \mathcal I $}
        \STATE Form the partially updated profile
        \[
        \widetilde K^{k,i}
        :=
        \bigl(K_1^{k+1},\dots,K_{i-1}^{k+1},K_i^k,\dots,K_N^k\bigr).
        \]

        \STATE \textbf{Policy evaluation:} Solve for $P_i^{k}\in\mathbb S^n_{++}$
        \begin{align*}
        0
        &=
        S_i(\widetilde K^{k,i})
        +P_i^{k}A_{\mathrm{cl}}(\widetilde K^{k,i})
        +A_{\mathrm{cl}}(\widetilde K^{k,i})^\top P_i^{k}
        \\
        &\quad
        +C_{\mathrm{cl}}(\widetilde K^{k,i})^\top
        P_i^{k}
        C_{\mathrm{cl}}(\widetilde K^{k,i}).
        \end{align*}

        \STATE \textbf{Policy improvement:} Update
        \begin{equation}
        \begin{aligned}
          K_i^{k+1}\nonumber
            :={}&
            \bigl(R_{ii}+D_i^\top P_i^{k}D_i\bigr)^{-1}
            \Bigl(
            B_i^\top P_i^{k} \\
            &\qquad
            + D_i^\top P_i^{k}
            \bigl(C-\sum_{j\in\mathcal I\setminus\{i\}} D_j\widetilde K_j^{k,i}\bigr)
            \Bigr)
        \end{aligned}
        \end{equation}
    \ENDFOR

    \STATE $k\gets k+1$
\UNTIL{$\displaystyle \sum_{i\in\mathcal I} \|K_i^k-K_i^{k-1}\|<\varepsilon$}

\RETURN $ K^k$
\end{algorithmic}
\end{algorithm}

\subsection{Sequential Iteration Mapping}
We now combine policy evaluation and policy improvement into a single player-wise update and compose these updates sequentially.
For a gain profile \(K\in\mathcal K_{\mathrm{MSS}}\) and a player \(i\in\mathcal I\), the value matrix \(P_i(K)\in\mathbb S^n_{++}\) is well defined by Lemma~\ref{lem:eval-Pi}. Evaluating \(P_i(K)\) and applying the improvement step \eqref{eq:Ui-stationarity} with \(P=P_i(K)\) yields the player-wise gain update \(\mathcal U_i(K)\), defined as the unique solution \(\widehat K_i\) of \eqref{eq:Ui-stationarity} at \(P=P_i(K)\),
\begin{equation}\label{eq:Ti-update}
\mathcal U_i(K):=\widehat K_i\big|_{P=P_i(K)}.
\end{equation}
\begin{definition}[Player-wise update mapping]\label{def:cascaded-Ti}
For each \(i\in\mathcal I\), define the player-wise update mapping
$
\mathcal T_i:\mathcal K_{\mathrm{MSS}}\to\mathcal K
$
by
\begin{equation}\label{eq:Ti-mapping}
\mathcal T_i(K):=
\bigl(K_1,\dots,K_{i-1},\mathcal U_i(K),K_{i+1},\dots,K_N\bigr).
\end{equation}
\end{definition}

Within each iteration, these player-wise updates are applied to one player at a time, in an arbitrary but fixed order, here \(1,\dots,N\).

Let \(K^k\in\mathcal K_{\mathrm{MSS}}\) denote the current iterate. The partially updated gain profiles generated within iteration \(k\) are defined recursively by
\begin{equation}\label{eq:stage-profile}
\widetilde K^{k,i}:=
\mathcal T_{i-1}\circ\cdots\circ \mathcal T_1(K^k),
\qquad i\in\mathcal I,
\end{equation}
where, by convention, \(\mathcal T_{i-1}\circ\cdots\circ \mathcal T_1\) denotes the identity mapping for \(i=1\).
Equivalently,
\begin{equation}\label{eq:cascaded-profile}
\widetilde K^{k,i}
=
\bigl(K_1^{k+1},\dots,K_{i-1}^{k+1},K_i^k,\dots,K_N^k\bigr),
\, i\in\mathcal I\setminus\{1\}.
\end{equation}
\begin{definition}[Sequential policy iteration mapping]\label{def:sequential-map}
Define the sequential policy-iteration mapping by
\begin{equation}\label{eq:sequential_update}
\mathcal T_{\mathrm{seq}}:=\mathcal T_N\circ \mathcal T_{N-1}\circ\cdots\circ \mathcal T_1.
\end{equation}
\end{definition}
Algorithm~\ref{alg:cascaded-PI} summarizes the resulting sequential policy-iteration scheme.

\subsection{Preservation of Mean-Square Stability}

Each \(\mathcal T_i\) is defined through the value matrix \(P_i(K)\) and therefore requires \(K\in\mathcal{K}_{\mathrm{MSS}}\). The composition \eqref{eq:sequential_update} is thus well posed only if every partially updated gain profile again lies in \(\mathcal{K}_{\mathrm{MSS}}\). The following theorem establishes that each player-wise update preserves mean-square stability, thereby confirming that Algorithm~\ref{alg:cascaded-PI} is well posed.
\begin{theorem}[Mean-square stability preservation]
\label{thm:mss-preservation-cascaded}
Fix \(k\in\mathbb N\) and let \(K^k\in\mathcal K_{\mathrm{MSS}}\). Let the stage
profiles \(\widetilde K^{k,i}\), \(i\in\mathcal I\), be defined by
\eqref{eq:stage-profile}, with \(\widetilde K^{k,1}=K^k\). Then
$
\widetilde K^{k,i}\in\mathcal K_{\mathrm{MSS}}$,
$ i\in\mathcal I.
$
Moreover,
$
K^{k+1}=\mathcal T_{\mathrm{seq}}(K^k)\in\mathcal K_{\mathrm{MSS}}.
$
\end{theorem}
\begin{proof}
For \(K\in\mathcal K_{\mathrm{MSS}}\), define the Lyapunov operator
\begin{equation}
\label{eq:proof-Lyapunov-operator}
\mathcal L_K(X)
:=
A_{\mathrm{cl}}(K)^\top X
+
X A_{\mathrm{cl}}(K)
+
C_{\mathrm{cl}}(K)^\top X C_{\mathrm{cl}}(K),
\end{equation}
with $ X\in\mathbb S^n$.
For notational convenience, set
$
A^{k,i}:=A_{\mathrm{cl}}(\widetilde K^{k,i})$,
$
C^{k,i}:=C_{\mathrm{cl}}(\widetilde K^{k,i})$ and
$
S_i^{k,i}:=S_i(\widetilde K^{k,i}).
$
In addition, define
$
\widetilde K^{k,N+1}:=K^{k+1}=\mathcal T_{\mathrm{seq}}(K^k).
$

Fix \(i\in\mathcal I\) and assume that
$
\widetilde K^{k,i}\in\mathcal K_{\mathrm{MSS}}.
$
By Lemma~\ref{lem:eval-Pi}, the generalized Lyapunov equation 
\begin{equation}
\label{eq:proof-stage-evaluation}
\mathcal L_{\widetilde K^{k,i}}(P_i^{k})+S_i^{k,i}=0
\end{equation}
admits a unique solution \(P_i^{k}\in\mathbb S_{++}^n\).

Let
$
\Delta K_i:=K_i^{k+1}-K_i^k.
$
Since only the \(i\)-th gain is updated, one has
\begin{equation}
\label{eq:proof-A-update}
A^{k,i+1}=A^{k,i}-B_i\Delta K_i
\end{equation}
and
\begin{equation}
\label{eq:proof-C-update}
C^{k,i+1}=C^{k,i}-D_i\Delta K_i.
\end{equation}
A direct expansion of the Lyapunov operator at \(\widetilde K^{k,i+1}\) yields
\begin{align}
\mathcal L_{\widetilde K^{k,i+1}}(P_i^{k})
&=
\mathcal L_{\widetilde K^{k,i}}(P_i^{k}) \notag \\
&\quad -\Delta K_i^\top\!\bigl(B_i^\top P_i^{k}+D_i^\top P_i^{k}C^{k,i}\bigr)
\notag\\
&\quad
-\bigl(P_i^{k}B_i+(C^{k,i})^\top P_i^{k}D_i\bigr)\Delta K_i \notag \\
&\quad +\Delta K_i^\top D_i^\top P_i^{k}D_i\,\Delta K_i.
\label{eq:proof-L-expansion}
\end{align}
Moreover, since only \(K_i\) changes,
\begin{align}
S_i^{k,i+1}
&=
S_i^{k,i}
-(K_i^k)^\top R_{ii}K_i^k
+(K_i^{k+1})^\top R_{ii}K_i^{k+1}.
\label{eq:proof-S-update}
\end{align}

Using the stationarity condition \eqref{eq:Ui-stationarity} for player \(i\) at
\(\widetilde K^{k,i}\), one obtains
\begin{equation}
\label{eq:proof-stationarity-rewrite}
B_i^\top P_i^{k}+D_i^\top P_i^{k}C^{k,i}
=
R_{ii}K_i^{k+1}+D_i^\top P_i^{k}D_i\,\Delta K_i.
\end{equation}
Substituting \eqref{eq:proof-stationarity-rewrite} into
\eqref{eq:proof-L-expansion} and combining the result with
\eqref{eq:proof-S-update} and the evaluation identity \eqref{eq:proof-stage-evaluation} yields
\begin{equation}
\label{eq:proof-completed-square}
\mathcal L_{\widetilde K^{k,i+1}}(P_i^{k})+S_i^{k,i+1}
=
-\Delta K_i^\top\bigl(R_{ii}+D_i^\top P_i^{k}D_i\bigr)\Delta K_i
\preceq0.
\end{equation}
Since \(Q_i\succ0\) and \(R_{ij}\succeq0\) for all \(j\in\mathcal I\), one has
$
S_i^{k,i+1}\succ0.
$
Therefore,
\begin{equation}
\label{eq:proof-Lyapunov-ineq}
\mathcal L_{\widetilde K^{k,i+1}}(P_i^{k})
=
-\,S_i^{k,i+1}
-\Delta K_i^\top\bigl(R_{ii}+D_i^\top P_i^{k}D_i\bigr)\Delta K_i
\prec0.
\end{equation}
Since \(P_i^{k}\succ0\), the MSS Lyapunov criterion \eqref{eq:Lyap-MSS}
implies that
$
\widetilde K^{k,i+1}\in\mathcal K_{\mathrm{MSS}}.
$
The claim follows by induction over \(i\in\mathcal I\).
\end{proof}

\subsection{Fixed-Point Characterization}
Theorem~\ref{thm:mss-preservation-cascaded} guarantees that the iteration remains on \(\mathcal K_{\mathrm{MSS}}\). It remains to characterize its fixed points and thereby confirm that, upon convergence, Algorithm~\ref{alg:cascaded-PI} returns a solution of the coupled equilibrium system \eqref{eq:equilibrium-eval}-\eqref{eq:equilibrium-improve}.

\begin{proposition}[Fixed-point characterization]
\label{prop:fixed-point-characterization}
A gain profile \(K^\star\in\mathcal K_{\mathrm{MSS}}\) satisfies
\(\mathcal T_{\mathrm{seq}}(K^\star)=K^\star\) if and only if it satisfies the coupled
equilibrium equations \eqref{eq:equilibrium-eval}--\eqref{eq:equilibrium-improve},
and hence if and only if it is a stabilizing linear feedback Nash equilibrium in
the sense of Definition~\ref{def:LFNE}. In that case,
\(\mathcal T_i(K^\star)=K^\star\) for every \(i\in\mathcal I\).
\end{proposition}

\begin{proof}
(Sufficiency) If \(K^\star\) satisfies \eqref{eq:equilibrium-eval}--\eqref{eq:equilibrium-improve}, then for each \(i\in\mathcal I\) the value matrix \(P_i(K^\star)\) solves \eqref{eq:equilibrium-eval} and \(K_i^\star\) satisfies the stationarity condition \eqref{eq:equilibrium-improve}; hence \(\mathcal U_i(K^\star)=K_i^\star\) and \(\mathcal T_i(K^\star)=K^\star\) for all \(i\in\mathcal I\) and therefore \(\mathcal T_{\mathrm{seq}}(K^\star)=K^\star\).

(Necessity) Suppose \(\mathcal T_{\mathrm{seq}}(K^\star)=K^\star\). By Definition~\ref{def:cascaded-Ti}, each map \(\mathcal T_j\) updates only block \(j\) and leaves the remaining blocks unchanged. Hence \(\mathcal T_1\) replaces \(K_1^\star\) by \(\mathcal U_1(K^\star)\); since block \(1\) is not modified at the later stages, \(\mathcal T_{\mathrm{seq}}(K^\star)=K^\star\) forces \(\mathcal U_1(K^\star)=K_1^\star\), so that the profile entering stage \(2\) is again \(K^\star\). Proceeding inductively, the profile entering each stage \(i\) equals \(K^\star\) and \(\mathcal U_i(K^\star)=K_i^\star\) for every \(i\in\mathcal I\), which is precisely the stationarity condition \eqref{eq:equilibrium-improve} with \(P^\star_i=P^\star_i(K^\star)\) solving \eqref{eq:equilibrium-eval}. Hence \(K^\star\) satisfies \eqref{eq:equilibrium-eval}--\eqref{eq:equilibrium-improve}, which are necessary and sufficient conditions for \(K^\star\) to be a stabilizing linear feedback Nash equilibrium~\cite{Lucas}.
\end{proof}

Together with the mean-square-stability preservation of Theorem~\ref{thm:mss-preservation-cascaded}, this establishes that Algorithm~\ref{alg:cascaded-PI} is a well-posed iteration on \(\mathcal K_{\mathrm{MSS}}\) whose fixed points solve the coupled equilibrium system \eqref{eq:equilibrium-eval}--\eqref{eq:equilibrium-improve}, thereby completing Contribution~1 and resolving Problem~\ref{prob:main}.

\section{Local Convergence of Sequential Policy Iteration}
\label{sec:local_convergence}

With $\mathcal T_{\mathrm{seq}}$ well defined on $\mathcal{K}_{\mathrm{MSS}}$ by Theorem~\ref{thm:mss-preservation-cascaded}, we turn to its behavior near a stabilizing linear feedback Nash equilibrium.
Such equilibria need not be unique, since the coupled equations may admit several solutions already in the deterministic special case contained in the present model \cite{engwerda2005,basar_olsder_1998}.
Convergence is therefore a local, equilibrium-dependent property.
We analyze the Fr\'echet derivative of the iteration map \eqref{eq:sequential_update}, which governs the first-order propagation of perturbations near an equilibrium and thus yields a criterion for local convergence \cite{OrtegaRheinboldt}.
The derivative identifies the sequential scheme and its decoupled counterpart, in which all players update simultaneously, as block Gauss--Seidel and block Jacobi iterations of a single operator, respectively. 
This explains the faster convergence reported for sequential updates \cite{nortmann_monti_sassano_mylvaganam2024_tac_iterative_datadriven_lq_games,chen_chen_lewis_xie_mcpi_nash_dg} and yields a verifiable condition under which the advantage holds.

\subsection{Linearization and Local Convergence Criterion}\label{sec:playerwise_sens}

To derive the Fr\'echet derivative of the full sequential iteration map \(\mathcal{T}_{\mathrm{seq}}\) from Definition~\ref{def:sequential-map}, we first characterize the local linearization of the player-wise update maps \(\mathcal{T}_i\) from Definition~\ref{def:cascaded-Ti}.
For fixed \(i\in\mathcal I\), define the policy-evaluation residual map
$
\mathcal{E}_i:\mathbb S^n\times\mathcal K\to\mathbb S^n
$
by
\begin{align}
\mathcal{E}_i(P,K)
&:=
A_{\mathrm{cl}}(K)^\top P
+
P A_{\mathrm{cl}}(K)
\label{eq:eval-map-local}
\\
&\quad
+
C_{\mathrm{cl}}(K)^\top P C_{\mathrm{cl}}(K)
+
S_i(K).
\nonumber
\end{align}
By Lemma~\ref{lem:eval-Pi}, for every \(K\in\mathcal K_{\mathrm{MSS}}\), the value matrix \(P_i(K)\) is characterized by
\begin{equation}
\label{eq:eval-map-local-characterization}
\mathcal{E}_i\bigl(P_i(K),K\bigr)=0.
\end{equation}

\begin{lemma}[Local smoothness of the evaluation map]
\label{lem:local-smooth-P}
Fix \(i\in\mathcal I\) and let \(K^\star\in\mathcal K_{\mathrm{MSS}}\). Then there exists an open neighborhood \(\mathcal N\subset\mathcal K_{\mathrm{MSS}}\) of \(K^\star\) such that the solution map
\begin{equation}
\label{eq:local-solution-map}
\mathcal N\ni K\mapsto P_i(K)\in\mathbb S^n
\end{equation}
defined through \eqref{eq:eval-map-local-characterization}, is continuously Fr\'echet differentiable on \(\mathcal N\). Moreover, the partial Fr\'echet derivative \(\mathcal{D}_P\mathcal{E}_i(K)\) of \(\mathcal{E}_i\) with respect to \(P\), evaluated at \(\bigl(P_i(K),K\bigr)\),
\begin{equation}
\label{eq:local-eval-derivative}
\mathcal{D}_P\mathcal{E}_i(K):\mathbb S^n\to\mathbb S^n,
\end{equation}
is invertible for every \(K\in\mathcal N\), and its inverse is uniformly bounded on \(\mathcal N\).
\end{lemma}
\begin{proof}
Since \(K^\star\in\mathcal K_{\mathrm{MSS}}\), the strict Lyapunov inequality \eqref{eq:Lyap-MSS} remains valid under sufficiently small perturbations of \(K\). Hence \(\mathcal K_{\mathrm{MSS}}\) is open, and there exists an open neighborhood \(\mathcal N\subset\mathcal K_{\mathrm{MSS}}\) of \(K^\star\). The map \(\mathcal{E}_i(P,K)\) is continuously Fr\'echet differentiable in \((P,K)\), and by Lemma~\ref{lem:eval-Pi} the derivative \(\mathcal{D}_P\mathcal{E}_i(K^\star)\) is precisely the Lyapunov operator \(\mathcal L_{K^\star}\) from \eqref{eq:proof-Lyapunov-operator}, which is invertible for \(K^\star\in\mathcal K_{\mathrm{MSS}}\); see, e.g., \cite{damm2004}. The claim therefore follows from the implicit function theorem, see e.g. \cite{OrtegaRheinboldt}. Since \(\mathcal{D}_P\mathcal{E}_i\) depends continuously on \(K\), invertibility and a bound on the inverse persist under sufficiently small perturbations of \(K\). Hence, after possibly shrinking \(\mathcal N\), the stated properties hold for all \(K\in\mathcal N\).
\end{proof}

For fixed \(i\in\mathcal I\), define the fixed-value stationarity residual map associated with \eqref{eq:Ui-stationarity} by
\begin{equation}
\label{eq:local-Gi-domain}
\mathcal{G}_i:\mathbb S^n\times\mathbb R^{m_i\times n}\times\mathcal K\to\mathbb R^{m_i\times n}
\end{equation}
by
\begin{align}
\mathcal{G}_i(P,\widehat K_i, K)
&:=
\bigl(R_{ii}+D_i^\top P D_i\bigr)\widehat K_i
\label{eq:local-Gi}
\\
&\quad
+\sum_{j\in\mathcal I\setminus\{i\}} D_i^\top P D_j\, K_j
- B_i^\top P
- D_i^\top P C .
\nonumber
\end{align}
Let \(K\in\mathcal N\), and define \(P_i:=P_i(K)\) and \(\widehat K_i:=\mathcal U_i(K)\), where \(\mathcal U_i\) is given by \eqref{eq:Ti-update}. By \eqref{eq:eval-map-local-characterization}, one has \(\mathcal{E}_i(P_i,K)=0\). Then, by the definition of \(\mathcal U_i\),
\begin{equation}
\label{eq:local-stationarity-at-K}
\mathcal{G}_i(P_i,\widehat K_i,K)=0.
\end{equation}
Here \(P_i\) is obtained by policy evaluation at \(K\), and \(\widehat K_i\) is then obtained by enforcing the fixed-value stationarity condition with \(P_i\) and \(K\) held fixed.

Next, in addition to \(\mathcal{D}_P\mathcal{E}_i(K)\) from Lemma~\ref{lem:local-smooth-P}, we write \(\mathcal{D}_K\mathcal{E}_i(K)\) for the partial Fr\'echet derivative of \(\mathcal{E}_i\) with respect to \(K\), evaluated at \(\bigl(P_i(K),K\bigr)\), and \(\mathcal{D}_P\mathcal{G}_i(K)\), \(\mathcal{D}_{\widehat K_i}\mathcal{G}_i(K)\), and \(\mathcal{D}_K\mathcal{G}_i(K)\) for the corresponding partial Fr\'echet derivatives of \(\mathcal{G}_i\), evaluated at \(\bigl(P_i(K),\mathcal U_i(K),K\bigr)\) (explicit expressions are given in Appendix~\ref{app:frechet-blocks}). Whenever no ambiguity arises, we omit the dependence on \(K\).

Throughout, for a Fr\'echet differentiable map \(F\) we write \(\mathcal{D}F(K)\) for its Fr\'echet derivative at the point \(K\), which is a bounded linear operator, and \(\mathcal{D}F(K)[\Delta K]\) for its action on a perturbation \(\Delta K\). The subscript \(\bigl[\,\cdot\,\bigr]_j\) extracts the \(j\)-th block component.

\begin{lemma}[Local linearization of a player-wise update]
\label{lem:reduced-playerwise}
Fix \(i\in\mathcal I\) and let \(K\in\mathcal N\). Then the player-wise update map \(\mathcal{T}_i\) from Definition~\ref{def:cascaded-Ti} is continuously Fr\'echet differentiable at \(K\), and the \(i\)-th block row of \(\mathcal{D}\mathcal{T}_i(K)\) equals the linear operator \(\mathcal{D}\mathcal U_i(K):\mathcal K\to\mathbb R^{m_i\times n}\) given by
\begin{equation}
\label{eq:DUi-operator}
\mathcal{D}\mathcal U_i(K)
=
-
\bigl(\mathcal{D}_{\widehat K_i}\mathcal{G}_i\bigr)^{-1}
\Bigl(
\mathcal{D}_K\mathcal{G}_i
-
\mathcal{D}_P\mathcal{G}_i(\mathcal{D}_P\mathcal{E}_i)^{-1}\mathcal{D}_K\mathcal{E}_i
\Bigr).
\end{equation}
All remaining block rows act as the identity. In particular, \(\mathcal{D}\mathcal{T}_i(K)\) has exactly one nontrivial block row, namely the \(i\)-th one.
\end{lemma}

\begin{proof}
By Lemma~\ref{lem:local-smooth-P}, the map
\(
K\mapsto P_i(K)
\)
is continuously Fr\'echet differentiable on \(\mathcal N\). Since \eqref{eq:eval-map-local-characterization} holds for every \(K\in\mathcal N\), differentiation of \eqref{eq:eval-map-local-characterization} at \(K\) in the direction \(\Delta K\in\mathcal K\) yields
\begin{equation}
\label{eq:proof-DP}
\mathcal{D}P_i(K)[\Delta K]
=
-(\mathcal{D}_P\mathcal{E}_i)^{-1}\mathcal{D}_K\mathcal{E}_i[\Delta K].
\end{equation}
Here \((\mathcal{D}_P\mathcal{E}_i)^{-1}\) is well defined by Lemma~\ref{lem:local-smooth-P}.
Next, differentiating \eqref{eq:local-stationarity-at-K} at \(K\) in the direction \(\Delta K\) gives
\begin{align}
\label{eq:proof-DG}
0 &= \mathcal{D}_P\mathcal{G}_i\bigl[\mathcal{D}P_i(K)[\Delta K]\bigr]
+
\mathcal{D}_{\widehat K_i}\mathcal{G}_i\bigl[\mathcal{D}\mathcal{U}_i(K)[\Delta K]\bigr]
\notag\\
&\qquad
+
\mathcal{D}_K\mathcal{G}_i[\Delta K].
\end{align}
Moreover, \(\mathcal{D}_{\widehat K_i}\mathcal{G}_i\) is invertible, since
\begin{equation}
\label{eq:proof-Ghat-invertible}
\mathcal{D}_{\widehat K_i}\mathcal{G}_i[H]
=
\bigl(R_{ii}+D_i^\top P_i(K)D_i\bigr)H,
\quad H\in\mathbb R^{m_i\times n},
\end{equation}
and \(R_{ii}+D_i^\top P_i(K)D_i\succ0\) by \(R_{ii}\succ0\) and Lemma~\ref{lem:eval-Pi}. Substituting \eqref{eq:proof-DP} into \eqref{eq:proof-DG} and solving for
\(\mathcal{D}\mathcal U_i(K)\) yields \eqref{eq:DUi-operator}.
Finally, by Definition~\ref{def:cascaded-Ti}, the map \(\mathcal{T}_i\) leaves all player
components \(j\in\mathcal I\setminus\{i\}\) unchanged and updates only the \(i\)-th component through
\(\mathcal U_i\); hence the \(i\)-th block row of \(\mathcal{D}\mathcal{T}_i(K)\) equals \(\mathcal{D}\mathcal U_i(K)\), while all other block rows act as the identity, so that \(\mathcal{D}\mathcal{T}_i(K)\) has exactly one nontrivial block row, namely the \(i\)-th one. Since
\(K\mapsto P_i(K)\) is continuously Fr\'echet differentiable on \(\mathcal N\)
and the residual maps \(\mathcal{E}_i\) and \(\mathcal{G}_i\) are continuously Fr\'echet
differentiable in their arguments, \(\mathcal{T}_i\) is continuously Fr\'echet
differentiable at \(K\).
\end{proof}

Since the full sequential iteration map from Definition~\ref{def:sequential-map} is obtained by composing the player-wise update maps from Definition~\ref{def:cascaded-Ti}, its Fr\'echet derivative is given by the chain rule; see, e.g., \cite{OrtegaRheinboldt}. Let \(K^\star\in\mathcal K_{\mathrm{MSS}}\) be a stabilizing linear feedback Nash equilibrium. By Proposition~\ref{prop:fixed-point-characterization}, \(\mathcal{T}_i(K^\star)=K^\star\) for all \(i\in\mathcal I\), so that every partial composition \(\mathcal{T}_{i-1}\circ\cdots\circ \mathcal{T}_1\) fixes \(K^\star\) as well. Hence all intermediate evaluation points of the chain rule coincide with \(K^\star\), and
\begin{equation}
\label{eq:full-jacobian-at-fixed-point}
\mathcal{D}\mathcal{T}_{\mathrm{seq}}(K^\star)
=
\mathcal{D}\mathcal{T}_N(K^\star)\,\mathcal{D}\mathcal{T}_{N-1}(K^\star)\cdots \mathcal{D}\mathcal{T}_1(K^\star).
\end{equation}

\begin{lemma}[Sufficient criterion for local convergence]
\label{lem:local-convergence}
Let \(K^\star\in\mathcal K_{\mathrm{MSS}}\) be a stabilizing linear feedback Nash equilibrium, and hence a solution of \eqref{eq:equilibrium-eval}--\eqref{eq:equilibrium-improve}. Suppose that the sequential policy-iteration map \(\mathcal{T}_{\mathrm{seq}}\) from Definition~\ref{def:sequential-map} is continuously Fr\'echet differentiable in a neighborhood of \(K^\star\). If
\begin{equation}
\label{eq:local-spectral-condition}
\rho\!\left(\mathcal{D}\mathcal{T}_{\mathrm{seq}}(K^\star)\right)<1,
\end{equation}
then there exists an equivalent norm \(\|\cdot\|_\ast\) on \(\mathcal K\) and a neighborhood \(\mathcal V\subset\mathcal K\) of \(K^\star\) such that, for every initial point \(K^0\in\mathcal V\), the iterates generated by repeated application of \(\mathcal{T}_{\mathrm{seq}}\) converge linearly to \(K^\star\) in the norm \(\|\cdot\|_\ast\).
\end{lemma}

\begin{proof}
Since \(\mathcal K\) is finite dimensional and \eqref{eq:local-spectral-condition} holds, there exists an equivalent norm \(\|\cdot\|_\ast\) on \(\mathcal K\) whose induced operator norm satisfies
\begin{equation}
\label{eq:proof-induced-norm}
\|\mathcal{D}\mathcal{T}_{\mathrm{seq}}(K^\star)\|_\ast<1;
\end{equation}
see, e.g., \cite{HornJohnson}. Because \(\mathcal{T}_{\mathrm{seq}}\) is continuously Fr\'echet differentiable in a neighborhood of \(K^\star\), the local contraction principle for nonlinear fixed-point iterations yields a constant \(q\in(0,1)\) and a radius \(r>0\) such that the contraction estimate
\begin{equation}
\label{eq:proof-local-contraction}
\|\mathcal{T}_{\mathrm{seq}}(K)-K^\star\|_\ast
\le q\,\|K-K^\star\|_\ast
\end{equation}
holds on the closed \(\|\cdot\|_\ast\)-ball \(\mathcal V:=\{K\in\mathcal K:\|K-K^\star\|_\ast\le r\}\); see, e.g., \cite{OrtegaRheinboldt,Deuflhard}. Since \(q<1\), the ball \(\mathcal V\) is forward invariant, since for \(K\in\mathcal V\) one has \(\|\mathcal{T}_{\mathrm{seq}}(K)-K^\star\|_\ast\le q\,\|K-K^\star\|_\ast\le q\,r<r\), so \(\mathcal{T}_{\mathrm{seq}}(K)\in\mathcal V\). Hence every iterate of \(\mathcal{T}_{\mathrm{seq}}\) with initial point \(K^0\in\mathcal V\) remains in \(\mathcal V\) and converges linearly to \(K^\star\).
\end{proof}
\begin{remark}[Vanishing local derivative]
If, in addition, \(\mathcal{D}\mathcal{T}_{\mathrm{seq}}(K^\star)=0\), then  the iteration converges locally superlinearly to \(K^\star\); see,
e.g., \cite{OrtegaRheinboldt,Deuflhard}.
\end{remark}

\subsection{Structural Decomposition of the Local Derivative}

Lemma~\ref{lem:local-convergence} shows that local convergence is governed by the Fr\'echet derivative \(\mathcal{D}\mathcal{T}_{\mathrm{seq}}(K^\star)\) through the spectral condition \eqref{eq:local-spectral-condition}.
So far, the sequential policy-iteration scheme has been stated for a fixed update order of the players. However, the sequential construction admits an additional degree of freedom, namely the order in which the player-wise updates are carried out.
We therefore analyze how the player-update order is reflected in the local derivative and thereby governs the local convergence behavior of sequential policy iteration.

For a permutation \(\pi\) of \(\mathcal I\), define the permuted sequential policy-iteration map by
\begin{equation}
\label{eq:permuted-sequential-map}
\mathcal{T}_{\mathrm{seq}}^{\pi}
:=
\mathcal{T}_{\pi(N)}\circ \mathcal{T}_{\pi(N-1)}\circ\cdots\circ \mathcal{T}_{\pi(1)}.
\end{equation}

\begin{proposition}[Order Decomposition]
\label{prop:order-decomposition}
Let \(K^\star\in\mathcal K_{\mathrm{MSS}}\) satisfy the coupled equilibrium equations
\eqref{eq:equilibrium-eval}--\eqref{eq:equilibrium-improve}, let \(\pi\) be a permutation of
\(\mathcal I\), and define
\begin{equation}
\label{eq:Ui-definition-permuted}
U_{\pi(\ell)}
:=
\mathcal{D}\mathcal{T}_{\pi(\ell)}(K^\star)-I,
\qquad \ell\in\mathcal I.
\end{equation}
Then the Fr\'echet derivative of the permuted sequential policy-iteration map
admits the decomposition
\begin{equation}
\label{eq:order-decomposition}
\mathcal{D}\mathcal{T}_{\mathrm{seq}}^{\pi}(K^\star)
=
\mathcal{D}\mathcal{T}_{\mathrm{inv}}+\Gamma_\pi,
\end{equation}
where
\begin{equation}
\label{eq:Jinv-definition}
\mathcal{D}\mathcal{T}_{\mathrm{inv}}
:=
I+\sum_{i\in\mathcal I} U_{\pi(i)}
\end{equation}
is the order-invariant part, and
\begin{equation}
\label{eq:Omega-definition}
\Gamma_\pi
:=
\sum_{m=2}^N
\ \sum_{1\le \ell_1<\cdots<\ell_m\le N}
U_{\pi(\ell_m)}\cdots U_{\pi(\ell_1)}
\end{equation}
is the order-dependent part, where, for each \(m\), the inner sum ranges over all strictly increasing index tuples \(1\le \ell_1<\cdots<\ell_m\le N\), equivalently over the \(m\)-element subsets of \(\mathcal I\).
\end{proposition}

\begin{proof}
Since \(K^\star\) satisfies \eqref{eq:equilibrium-eval}--\eqref{eq:equilibrium-improve},
one has \(\mathcal{T}_i(K^\star)=K^\star\) for all \(i\in\mathcal I\). Hence, by the chain rule,
\begin{equation}
\label{eq:proof-permuted-product}
\mathcal{D}\mathcal{T}_{\mathrm{seq}}^{\pi}(K^\star)
=
\mathcal{D}\mathcal{T}_{\pi(N)}(K^\star)\cdots \mathcal{D}\mathcal{T}_{\pi(1)}(K^\star).
\end{equation}
Using \eqref{eq:Ui-definition-permuted}, this becomes
\begin{equation}
\label{eq:proof-product-expansion}
\mathcal{D}\mathcal{T}_{\mathrm{seq}}^{\pi}(K^\star)
=
(I+U_{\pi(N)})\cdots(I+U_{\pi(1)}).
\end{equation}
Expanding the product yields the identity term \(I\), the first-order sum
\(\sum_{i\in\mathcal I} U_{\pi(i)}\), and all ordered products of length at least two.
The first two contributions form \(\mathcal{D}\mathcal{T}_{\mathrm{inv}}\), while the remaining terms are
collected in \(\Gamma_\pi\). This proves \eqref{eq:order-decomposition}.
\end{proof}

\subsection{Effects of the Update Order}

Proposition~\ref{prop:order-decomposition} identifies \(\Gamma_\pi\) as the sole source of order dependence in \(\mathcal{D}\mathcal{T}_{\mathrm{seq}}^{\pi}(K^\star)\). We now determine which properties of the sequential iteration are governed by this term and which remain invariant under permutations.

\begin{lemma}[Order invariance of fixed points]
\label{lem:order-invariant-fixpoint}
Let \(\pi\) be a permutation of \(\mathcal I\), and let \(\mathcal{T}_{\mathrm{seq}}^{\pi}\) be defined by \eqref{eq:permuted-sequential-map}. If \(K^\star\in\mathcal K_{\mathrm{MSS}}\) satisfies the coupled equilibrium equations \eqref{eq:equilibrium-eval}--\eqref{eq:equilibrium-improve}, then
\begin{equation}
\label{eq:order-invariant-fixpoint}
\mathcal{T}_{\mathrm{seq}}^{\pi}(K^\star)=K^\star
\end{equation}
for every permutation \(\pi\).
\end{lemma}

\begin{proof}
If \(K^\star\in\mathcal K_{\mathrm{MSS}}\) satisfies
\eqref{eq:equilibrium-eval}--\eqref{eq:equilibrium-improve}, then
\(\mathcal{T}_i(K^\star)=K^\star\) for all \(i\in\mathcal I\) by
Definition~\ref{def:cascaded-Ti}. Since \(\mathcal{T}_{\mathrm{seq}}^{\pi}\) in
\eqref{eq:permuted-sequential-map} is a composition of the player-wise
update maps \(\mathcal{T}_i\), each of which fixes \(K^\star\), it follows
that
$
\mathcal{T}_{\mathrm{seq}}^{\pi}(K^\star)=K^\star,
$ which is exactly \eqref{eq:order-invariant-fixpoint}.
\end{proof}

While the fixed-point set is thus permutation-invariant, the spectral radius of \(\mathcal{D}\mathcal{T}_{\mathrm{seq}}^{\pi}(K^\star)\), which governs the local convergence rate by Lemma~\ref{lem:local-convergence}, is in general not.
The following result establishes the precise extent of this spectral order dependence.

\begin{lemma}[Spectral order dependence]
\label{lem:spectral-order-dependence}
Let \(K^\star\in\mathcal K_{\mathrm{MSS}}\) satisfy
\eqref{eq:equilibrium-eval}--\eqref{eq:equilibrium-improve} and let \(\pi\) be a
permutation of \(\mathcal I\).
\begin{enumerate}
\renewcommand{\labelenumi}{(\roman{enumi})}
\item If \(\tilde\pi\) is a cyclic shift of \(\pi\), then
\(\mathcal{D}\mathcal{T}_{\mathrm{seq}}^{\pi}(K^\star)\) and
\(\mathcal{D}\mathcal{T}_{\mathrm{seq}}^{\tilde\pi}(K^\star)\) have the same spectrum, and in
particular
\begin{equation}
\label{eq:cyclic-spectral-radius}
\rho\!\left(\mathcal{D}\mathcal{T}_{\mathrm{seq}}^{\pi}(K^\star)\right)
=
\rho\!\left(\mathcal{D}\mathcal{T}_{\mathrm{seq}}^{\tilde\pi}(K^\star)\right).
\end{equation}
\item For \(N\ge3\), there exist a game \eqref{eq:sde-N}--\eqref{eq:cost}, a
stabilizing linear feedback Nash equilibrium \(K^\star\), and permutations
\(\pi,\tilde\pi\) not related by a cyclic shift such that
\(\rho(\mathcal{D}\mathcal{T}_{\mathrm{seq}}^{\pi}(K^\star))<1<
\rho(\mathcal{D}\mathcal{T}_{\mathrm{seq}}^{\tilde\pi}(K^\star))\).
\end{enumerate}
\end{lemma}

\begin{proof}
Claim~(i) follows from \eqref{eq:full-jacobian-at-fixed-point}, since
a cyclic shift of \(\pi\) corresponds to a cyclic permutation of the factors
\(\mathcal{D}\mathcal{T}_{\pi(\ell)}(K^\star)\) in the ordered product. For square matrices \(A\)
and \(B\) of the same size,
\begin{equation}
\label{eq:AB-BA-spectrum}
\det(\lambda I-AB)=\det(\lambda I-BA),
\end{equation}
and repeated application of \eqref{eq:AB-BA-spectrum} shows that any cyclic
permutation of the factors leaves the characteristic polynomial unchanged,
yielding \eqref{eq:cyclic-spectral-radius}.

For claim~(ii), let \(N=3\), \(n=1\), \(m_i=1\), \(A=0\), \(B_i=1\), \(C=0\),
\(D_i=0\), \(Q_i=\tfrac12\), and let the control weights be
\(R_{11}=R_{22}=R_{33}=1\), \(R_{13}=R_{21}=R_{32}=\tfrac92\) and
\(R_{12}=R_{23}=R_{31}=0\). Then \(K^\star=(1,1,1)\) with \(P_i^\star=1\) solves
\eqref{eq:equilibrium-eval}--\eqref{eq:equilibrium-improve} and
\(A_{\mathrm{cl}}(K^\star)=-3\), and \eqref{eq:DUi-operator} gives
\[
\mathcal{D}\mathcal{T}_{\mathrm{dec}}(K^\star)=
\begin{bmatrix}
0 & -\tfrac13 & \tfrac76\\[2pt]
\tfrac76 & 0 & -\tfrac13\\[2pt]
-\tfrac13 & \tfrac76 & 0
\end{bmatrix}.
\]
The six update orders split into the two cyclic classes of~(i), with
\(\rho=\tfrac{8}{27}<1\) and \(\rho=(65+23\sqrt{43})/108>1\). See also
Section~\ref{sec:numerical_experiments} for a stochastic three-player instance.
\end{proof}

By Lemma~\ref{lem:local-convergence}, the update order therefore governs the
local convergence rate at a given stabilizing fixed point, and can decide whether
that fixed point is locally attracting at all. 
In the example the game and the equilibrium are held fixed and only the composition order changes. The example is deterministic, so the effect is not specific to the stochastic setting.

\begin{remark}[Two-player spectral invariance]
\label{rem:two-player-spectral-invariance}
For \(N=2\), the two possible update orders \(\pi=(1,2)\) and
\(\pi=(2,1)\) are related by a cyclic shift.
Lemma~\ref{lem:spectral-order-dependence} therefore implies
\begin{equation}
\label{eq:two-player-spectral-radius}
\rho\!\left(\mathcal{D}\mathcal{T}_{\mathrm{seq}}^{(1,2)}(K^\star)\right)
=
\rho\!\left(\mathcal{D}\mathcal{T}_{\mathrm{seq}}^{(2,1)}(K^\star)\right),
\end{equation}
so that the sufficient spectral convergence condition
\eqref{eq:local-spectral-condition} of Lemma~\ref{lem:local-convergence} is
insensitive to the update order for \(N=2\).
\end{remark}

\subsection{Relation between Sequential and Decoupled Policy Iteration}

The preceding subsection characterized the order-dependent part \(\Gamma_\pi\).
We now show that the order-invariant part \(\mathcal{D}\mathcal{T}_{\mathrm{inv}}\) is the derivative of the decoupled scheme, in which every player is updated from the common gain profile \(K^{k}\) rather than from the partially updated profile \(\widetilde K^{k,i}\), and identify conditions under which the sequential structure is preferable.

On the neighborhood \(\mathcal N\) from Lemma~\ref{lem:local-smooth-P}, define the decoupled policy-iteration map by
\begin{equation}
\label{eq:decoupled-map}
\mathcal{T}_{\mathrm{dec}}(K)
:=
\bigl(\mathcal U_1(K),\dots,\mathcal U_N(K)\bigr),
\quad K\in\mathcal N,
\end{equation}
where \(\mathcal U_i\) is given by \eqref{eq:Ti-update} for each \(i\in\mathcal I\).

\begin{proposition}[Decoupled  PI derivative]
\label{prop:decoupled-jacobian}
Let \(K^\star\in\mathcal K_{\mathrm{MSS}}\) satisfy
\eqref{eq:equilibrium-eval}--\eqref{eq:equilibrium-improve}, let
\(\mathcal{T}_{\mathrm{dec}}\) be defined by \eqref{eq:decoupled-map}, and let
\(\mathcal{D}\mathcal{T}_{\mathrm{inv}}\) be given by \eqref{eq:Jinv-definition}. Then
\begin{equation}
\label{eq:decoupled-jacobian}
\mathcal{D}\mathcal{T}_{\mathrm{dec}}(K^\star)=\mathcal{D}\mathcal{T}_{\mathrm{inv}}.
\end{equation}
\end{proposition}

\begin{proof}
By Lemma~\ref{lem:reduced-playerwise}, each gain-update map
\(\mathcal U_i\) given by \eqref{eq:Ti-update} is continuously Fr\'echet
differentiable on \(\mathcal N\); in particular, the map
\(\mathcal{T}_{\mathrm{dec}}\) in \eqref{eq:decoupled-map} is Fr\'echet differentiable at
\(K^\star\). Moreover, by \eqref{eq:decoupled-map}, the \(i\)-th block of
\(\mathcal{T}_{\mathrm{dec}}(K)\) is \(\mathcal U_i(K)\). Hence, for every
\(\Delta K\in\mathcal K\),
\begin{equation}
\label{eq:proof-dec-block}
\bigl[\mathcal{D}\mathcal{T}_{\mathrm{dec}}(K^\star)[\Delta K]\bigr]_i
=
\mathcal{D}\mathcal U_i(K^\star)[\Delta K], \quad i \in \mathcal I.
\end{equation}

Now let \(U_i=\mathcal{D}\mathcal{T}_i(K^\star)-I\). By Lemma~\ref{lem:reduced-playerwise}, the \(i\)-th block row of \(\mathcal{D}\mathcal{T}_i(K^\star)\) equals \(\mathcal{D}\mathcal U_i(K^\star)\) and all other block rows are the identity; hence, writing \(\Delta\widehat K_i:=\mathcal{D}\mathcal U_i(K^\star)[\Delta K]\), one has
$
\bigl[U_i[\Delta K]\bigr]_j=0$, $
j\in\mathcal I\setminus\{i\},
$
and
$
\bigl[U_i[\Delta K]\bigr]_i
=
\Delta\widehat K_i-\Delta K_i.
$
Therefore, for every
\(\Delta K\in\mathcal K\),
\begin{equation}
\label{eq:proof-Jinv-block}
\Bigl[\Bigl(I+\sum_{i\in\mathcal I} U_i\Bigr)[\Delta K]\Bigr]_i
=
\Delta K_i+\bigl(\Delta\widehat K_i-\Delta K_i\bigr)
=
\Delta\widehat K_i.
\end{equation}
Comparison of
\eqref{eq:proof-dec-block} and \eqref{eq:proof-Jinv-block} thus yields
$
\mathcal{D}\mathcal{T}_{\mathrm{dec}}(K^\star)=I+\sum_{i\in\mathcal I} U_i.
$
By \eqref{eq:Jinv-definition}, this is exactly \(\mathcal{D}\mathcal{T}_{\mathrm{inv}}\).
\end{proof}

Since \(K^\star\) satisfies \eqref{eq:equilibrium-eval}--\eqref{eq:equilibrium-improve}, one has \(\mathcal{T}_{\mathrm{dec}}(K^\star)=K^\star\) by \eqref{eq:decoupled-map}. Moreover, the proof of Proposition~\ref{prop:decoupled-jacobian} shows that \(\mathcal{T}_{\mathrm{dec}}\) is continuously Fr\'echet differentiable at \(K^\star\). Hence the same argument as in Lemma~\ref{lem:local-convergence} applies to \(\mathcal{T}_{\mathrm{dec}}\). If
$
\rho\!\left(\mathcal{D}\mathcal{T}_{\mathrm{dec}}(K^\star)\right)<1,
$
or equivalently, by \eqref{eq:decoupled-jacobian},
$
\rho(\mathcal{D}\mathcal{T}_{\mathrm{inv}})<1,
$
then decoupled policy iteration converges locally linearly to \(K^\star\).

To relate the two schemes, we express the sequential local derivative as a block Gauss--Seidel--type iteration of the decoupled derivative \(\mathcal{D}\mathcal{T}_{\mathrm{inv}}\) (Proposition~\ref{prop:triangular-representation}). 

For a permutation \(\pi\) of \(\mathcal I\), define the \(\pi\)-ordered block space by
$
\mathcal K^\pi
:=
\prod_{i\in\mathcal I}\mathbb R^{m_{\pi(i)}\times n}.
$

\begin{definition}[Block-permutation operator]
\label{def:block-permutation}
Let \(\pi\) be a permutation of \(\mathcal I\). The block-permutation operator
\begin{equation}
\label{eq:block-permutation-operator}
\mathcal{P}_\pi:\mathcal K\to\mathcal K^\pi
\end{equation}
is defined by
\begin{equation}
\label{eq:block-permutation-action}
\mathcal{P}_\pi(\Delta K_1,\dots,\Delta K_N)
:=
(\Delta K_{\pi(1)},\dots,\Delta K_{\pi(N)}).
\end{equation}
Its inverse is the block-permutation operator
\begin{equation}
\label{eq:block-permutation-inverse}
\mathcal{P}_\pi^{-1}:\mathcal K^\pi\to\mathcal K,
\end{equation}
which reorders the blocks according to the inverse permutation \(\pi^{-1}\).
\end{definition}

For every linear operator \(A:\mathcal K\to\mathcal K\), define its
\(\pi\)-ordered representation by
$
\mathcal{P}_\pi A \mathcal{P}_\pi^{-1}:\mathcal K^\pi\to\mathcal K^\pi.
$

Define
\begin{equation}
\label{eq:Mpi-definition}
\mathcal{D}\mathcal{T}_{\mathrm{dec}}^{\pi}:=\mathcal{P}_\pi \mathcal{D}\mathcal{T}_{\mathrm{dec}}(K^\star)\mathcal{P}_\pi^{-1}.
\end{equation}
By \eqref{eq:block-permutation-inverse}, \(\mathcal{P}_\pi\) is invertible, so that \(\mathcal{D}\mathcal{T}_{\mathrm{dec}}^{\pi}\) is similar to \(\mathcal{D}\mathcal{T}_{\mathrm{dec}}(K^\star)\). Since similar operators share the same spectrum,
$
\rho(\mathcal{D}\mathcal{T}_{\mathrm{dec}}^{\pi})=\rho\!\left(\mathcal{D}\mathcal{T}_{\mathrm{dec}}(K^\star)\right),
$
and by \eqref{eq:decoupled-jacobian},
$
\mathcal{D}\mathcal{T}_{\mathrm{dec}}^{\pi}=\mathcal{P}_\pi \mathcal{D}\mathcal{T}_{\mathrm{inv}}\mathcal{P}_\pi^{-1}.
$

Write the block splitting of \(\mathcal{D}\mathcal{T}_{\mathrm{dec}}^{\pi}\) as
\begin{equation}
\label{eq:Mpi-splitting}
\mathcal{D}\mathcal{T}_{\mathrm{dec}}^{\pi}=L_\pi+D_\pi+R_\pi,
\end{equation}
where \(L_\pi\) is strictly block lower triangular, \(D_\pi\) is block
diagonal, and \(R_\pi\) is strictly block upper triangular.

\begin{proposition}[Triangular block representation]
\label{prop:triangular-representation}
Let \(\mathcal{T}_{\mathrm{seq}}^{\pi}\) be defined by \eqref{eq:permuted-sequential-map},
let \(\mathcal{P}_\pi\) be the block-permutation operator from
Definition~\ref{def:block-permutation}, let \(\mathcal{D}\mathcal{T}_{\mathrm{dec}}^{\pi}\) be given by
\eqref{eq:Mpi-definition}, and let \eqref{eq:Mpi-splitting} be the block
splitting of \(\mathcal{D}\mathcal{T}_{\mathrm{dec}}^{\pi}\). Then
\begin{equation}
\label{eq:operator-splitting}
\mathcal{P}_\pi \mathcal{D}\mathcal{T}_{\mathrm{seq}}^{\pi}(K^\star) \mathcal{P}_\pi^{-1}
=
(I-L_\pi)^{-1}(D_\pi+R_\pi).
\end{equation}
\end{proposition}

\begin{proof}
For \(i\in\mathcal I\), set
$
V_i:=\mathcal{P}_\pi U_{\pi(i)}\mathcal{P}_\pi^{-1}.
$
By Lemma~\ref{lem:reduced-playerwise}, each \(U_{\pi(i)}\) has
exactly one nontrivial block row, namely the \(\pi(i)\)-th one. Hence each
\(V_i\) has exactly one nontrivial block row, namely the \(i\)-th one, and the
\(i\)-th block row of \(I+V_i\) coincides with the \(i\)-th block row of
\(\mathcal{D}\mathcal{T}_{\mathrm{dec}}^{\pi}\).

Now fix \(\Delta K^{\pi}\in\mathcal K^\pi\), and define
$
\Delta K^{\pi,(0)}:=\Delta K^{\pi}$, and
$
\Delta K^{\pi,(i)}:=(I+V_i)\Delta K^{\pi,(i-1)}$, for
$ i \in \mathcal I.
$
Then, by \eqref{eq:permuted-sequential-map} and the chain rule at \(K^\star\),
$
\Delta K^{\pi,(N)}=\mathcal{P}_\pi \mathcal{D}\mathcal{T}_{\mathrm{seq}}^{\pi}(K^\star)\mathcal{P}_\pi^{-1}\Delta K^{\pi}.
$
Since \(I+V_i\) differs from the identity only in block row \(i\), one has
\begin{equation}
(\Delta K^{\pi,(i)})_j=(\Delta K^{\pi,(i-1)})_j, \qquad j \in \mathcal I\setminus\{i\},
\end{equation}

and
\begin{equation}
(\Delta K^{\pi,(i)})_i=\sum_{j\in\mathcal I} (\mathcal{D}\mathcal{T}_{\mathrm{dec}}^{\pi})_{ij}(\Delta K^{\pi,(i-1)})_j.
\end{equation}

For \(j<i\), the \(j\)-th block is unchanged after step \(j\), hence
$
(\Delta K^{\pi,(i-1)})_j=(\Delta K^{\pi,(N)})_j.
$
For \(j\ge i\), the \(j\)-th block has not yet been updated at step \(i-1\),
hence
$
(\Delta K^{\pi,(i-1)})_j=\Delta K^{\pi}_j.
$
Therefore,
\begin{equation}
(\Delta K^{\pi,(N)})_i
=
\sum_{j<i}(\mathcal{D}\mathcal{T}_{\mathrm{dec}}^{\pi})_{ij}(\Delta K^{\pi,(N)})_j
+
\sum_{j\ge i}(\mathcal{D}\mathcal{T}_{\mathrm{dec}}^{\pi})_{ij}\Delta K^{\pi}_j,
\end{equation}

 for $ i\in\mathcal I$. Equivalently, using \eqref{eq:Mpi-splitting}, we get
\begin{equation}
\Delta K^{\pi,(N)}=L_\pi \Delta K^{\pi,(N)}+(D_\pi+R_\pi)\Delta K^{\pi}.
\end{equation}

Since \(L_\pi\) is strictly block lower triangular, \(I-L_\pi\) is invertible,
and thus
\begin{equation}
\Delta K^{\pi,(N)}=(I-L_\pi)^{-1}(D_\pi+R_\pi)\Delta K^{\pi}.
\end{equation}

Because this holds for every \(\Delta K^{\pi}\in\mathcal K^\pi\), \eqref{eq:operator-splitting}
follows.
\end{proof}

\begin{remark}[Jacobi and Gauss--Seidel structure]
Representation \eqref{eq:operator-splitting} places the update order within the classical theory of matrix splittings \cite{varga_2000,householder_1956}. The local dynamics of the decoupled scheme is the block-Jacobi iteration of the splitting \eqref{eq:Mpi-splitting} of \(\mathcal{D}\mathcal{T}_{\mathrm{inv}}\), updating all player blocks simultaneously, whereas the sequential scheme is the corresponding block Gauss--Seidel iteration, updating the blocks in turn, each using the already-updated ones. The local effect of the update order is therefore exactly the difference between Jacobi- and Gauss--Seidel-type iterations of the same operator, which both interprets the ordering phenomenon and provides the structural basis for the comparison below.
The representation uses only that each player-wise update modifies a single gain block, and therefore applies to player-wise policy iteration beyond the present class.
\end{remark}

Whereas \eqref{eq:local-spectral-condition} only guarantees the existence of a contracting norm, we compare the two schemes in an explicit, checkable weighted block-\(\ell_\infty\) norm. The weights provide the freedom to tighten the resulting contraction bound toward the spectral radius \cite{HornJohnson}. For \(\Delta K^{\pi}=(\Delta K^{\pi}_1,\dots,\Delta K^{\pi}_N)\in\mathcal K^\pi\) and weights \(w_1,\dots,w_N>0\), define
\begin{equation}
\label{eq:weighted-block-norm}
\|\Delta K^{\pi}\|_{w,\infty}
:=
\max_{1\le i\le N}\frac{\|\Delta K^{\pi}_i\|}{w_i}.
\end{equation}
This induces the transported norm on \(\mathcal K\),
\begin{equation}
\label{eq:transported-block-norm}
\|\Delta K\|_{\pi,w,\infty}
:=
\|\mathcal{P}_\pi \Delta K\|_{w,\infty},
\qquad \Delta K\in\mathcal K.
\end{equation}
Then
\begin{equation}
\label{eq:block-row-operator-norm}
\|\mathcal{D}\mathcal{T}_{\mathrm{dec}}^{\pi}\|_{w,\infty}
=
\max_{i\in\mathcal I}
\sup_{\|\Delta K^{\pi}\|_{w,\infty}\le 1}
\frac{\bigl\|\sum_{j\in\mathcal I}(\mathcal{D}\mathcal{T}_{\mathrm{dec}}^{\pi})_{ij}\Delta K^{\pi}_j\bigr\|}{w_i}.
\end{equation}
Choosing \(w\) as a dominant eigenvector of the nonnegative matrix with entries \(\|(\mathcal{D}\mathcal{T}_{\mathrm{dec}}^{\pi})_{ij}\|\) renders the resulting block bound minimal.

\begin{theorem}[Weighted block-norm comparison]
\label{thm:seq-no-worse}
Let \(K^\star\in\mathcal K_{\mathrm{MSS}}\) satisfy
\eqref{eq:equilibrium-eval}--\eqref{eq:equilibrium-improve}, let
\(\mathcal{T}_{\mathrm{seq}}^{\pi}\) be defined by \eqref{eq:permuted-sequential-map}, let
\(\mathcal{T}_{\mathrm{dec}}\) be defined by \eqref{eq:decoupled-map}, let \(\mathcal{P}_\pi\) be
the block-permutation operator from Definition~\ref{def:block-permutation}, and
let the norms \(\|\cdot\|_{w,\infty}\) on \(\mathcal K^\pi\) and
\(\|\cdot\|_{\pi,w,\infty}\) on \(\mathcal K\) be given by
\eqref{eq:weighted-block-norm} and \eqref{eq:transported-block-norm},
respectively. Assume that
\begin{equation}
\label{eq:dec-assumption}
\bigl\|\mathcal{D}\mathcal{T}_{\mathrm{dec}}(K^\star)\bigr\|_{\pi,w,\infty}\le 1.
\end{equation}
Then
\begin{equation}
\label{eq:seq-no-worse}
\bigl\|\mathcal{D}\mathcal{T}_{\mathrm{seq}}^{\pi}(K^\star)\bigr\|_{\pi,w,\infty}
\le
\bigl\|\mathcal{D}\mathcal{T}_{\mathrm{dec}}(K^\star)\bigr\|_{\pi,w,\infty}.
\end{equation}
Thus, in the norm \(\|\cdot\|_{\pi,w,\infty}\), the local linear convergence rate for sequential policy iteration is at least as good as that for decoupled policy iteration.
\end{theorem}

\begin{proof}
By \eqref{eq:transported-block-norm} and \eqref{eq:Mpi-definition}, one has
\begin{equation}
\label{eq:proof-dec-norm-equality}
\bigl\|\mathcal{D}\mathcal{T}_{\mathrm{dec}}(K^\star)\bigr\|_{\pi,w,\infty}
=
\|\mathcal{D}\mathcal{T}_{\mathrm{dec}}^{\pi}\|_{w,\infty}.
\end{equation}
Equations \eqref{eq:dec-assumption} and \eqref{eq:proof-dec-norm-equality} imply \(\|\mathcal{D}\mathcal{T}_{\mathrm{dec}}^{\pi}\|_{w,\infty}\le 1\).
Now fix \(\Delta K^{\pi}\in\mathcal K^\pi\) with \(\|\Delta K^{\pi}\|_{w,\infty}\le 1\), and set
\begin{equation}
\label{eq:proof-y-definition}
\Delta K^{\pi,+}:=(I-L_\pi)^{-1}(D_\pi+R_\pi)\Delta K^{\pi}.
\end{equation}
By Proposition~\ref{prop:triangular-representation},
\begin{equation}
\label{eq:proof-y-representation}
\Delta K^{\pi,+}=\mathcal{P}_\pi \mathcal{D}\mathcal{T}_{\mathrm{seq}}^{\pi}(K^\star)\mathcal{P}_\pi^{-1}\Delta K^{\pi}.
\end{equation}
We show by induction on \(i\in\mathcal I\) that
\begin{equation}
\label{eq:proof-induction-claim}
\frac{\|\Delta K^{\pi,+}_i\|}{w_i}\le \|\mathcal{D}\mathcal{T}_{\mathrm{dec}}^{\pi}\|_{w,\infty}.
\end{equation}

For \(i=1\), equations \eqref{eq:operator-splitting} and \eqref{eq:Mpi-splitting} give
\begin{equation}
\label{eq:proof-y1-row}
\Delta K^{\pi,+}_1=\sum_{j\in\mathcal I}(\mathcal{D}\mathcal{T}_{\mathrm{dec}}^{\pi})_{1j}\Delta K^{\pi}_j.
\end{equation}
Hence, by \eqref{eq:block-row-operator-norm},
\begin{equation}
\label{eq:proof-y1-bound}
\frac{\|\Delta K^{\pi,+}_1\|}{w_1}\le \|\mathcal{D}\mathcal{T}_{\mathrm{dec}}^{\pi}\|_{w,\infty}.
\end{equation}

Now let \(i\ge 2\), and assume that \(\|\Delta K^{\pi,+}_j\|/w_j\le \|\mathcal{D}\mathcal{T}_{\mathrm{dec}}^{\pi}\|_{w,\infty}\) for all \(j<i\). Define \(\Delta\widetilde{K}^{\pi,(i)}:=(\Delta K^{\pi,+}_1,\dots,\Delta K^{\pi,+}_{i-1},\Delta K^{\pi}_i,\dots,\Delta K^{\pi}_N)\). 

Since \(\|\mathcal{D}\mathcal{T}_{\mathrm{dec}}^{\pi}\|_{w,\infty}\le 1\), the induction hypothesis gives \(\|\Delta K^{\pi,+}_j\|/w_j\le 1\) for \(j<i\); together with \(\|\Delta K^{\pi}\|_{w,\infty}\le 1\), this implies
\begin{equation}
\label{eq:proof-tilde-z-bound}
\|\Delta\widetilde{K}^{\pi,(i)}\|_{w,\infty}\le 1.
\end{equation}
Moreover, equations \eqref{eq:operator-splitting} and \eqref{eq:Mpi-splitting} yield
\begin{align}
\label{eq:proof-yi-row}
\Delta K^{\pi,+}_i
&=
\sum_{j<i}(\mathcal{D}\mathcal{T}_{\mathrm{dec}}^{\pi})_{ij}\Delta K^{\pi,+}_j
+
\sum_{j\ge i}(\mathcal{D}\mathcal{T}_{\mathrm{dec}}^{\pi})_{ij}\Delta K^{\pi}_j \nonumber \\ 
&= 
\sum_{j\in\mathcal I}(\mathcal{D}\mathcal{T}_{\mathrm{dec}}^{\pi})_{ij}\Delta\widetilde{K}^{\pi,(i)}_j.
\end{align}
Therefore, by \eqref{eq:block-row-operator-norm} and \eqref{eq:proof-tilde-z-bound},
\begin{equation}
\label{eq:proof-yi-bound}
\frac{\|\Delta K^{\pi,+}_i\|}{w_i}\le \|\mathcal{D}\mathcal{T}_{\mathrm{dec}}^{\pi}\|_{w,\infty}.
\end{equation}
This proves \eqref{eq:proof-induction-claim} for all \(i\in\mathcal I\). Taking the maximum over \(i\) gives
\begin{equation}
\label{eq:proof-y-bound}
\|\Delta K^{\pi,+}\|_{w,\infty}\le \|\mathcal{D}\mathcal{T}_{\mathrm{dec}}^{\pi}\|_{w,\infty}.
\end{equation}

By \eqref{eq:proof-y-representation}, \eqref{eq:proof-y-bound}, and the arbitrariness of \(\Delta K^{\pi}\) with \(\|\Delta K^{\pi}\|_{w,\infty}\le 1\), we obtain
\begin{equation}
\label{eq:proof-conjugated-bound}
\bigl\|\mathcal{P}_\pi \mathcal{D}\mathcal{T}_{\mathrm{seq}}^{\pi}(K^\star)\mathcal{P}_\pi^{-1}\bigr\|_{w,\infty}
\le
\|\mathcal{D}\mathcal{T}_{\mathrm{dec}}^{\pi}\|_{w,\infty}.
\end{equation}
Using \eqref{eq:transported-block-norm} once more together with \eqref{eq:proof-dec-norm-equality}, we obtain \eqref{eq:seq-no-worse}.
\end{proof}

Wherever the decoupled scheme is non-expansive in \(\|\cdot\|_{\pi,w,\infty}\),
sequential updating therefore contracts at least as strongly.

\begin{corollary}[Two-player spectral relation]
\label{cor:two-player-squaring}
Let \(N=2\) and let \(K^\star\in\mathcal K_{\mathrm{MSS}}\) satisfy
\eqref{eq:equilibrium-eval}--\eqref{eq:equilibrium-improve}. Then, for either
update order,
\begin{equation}
\label{eq:two-player-squaring}
\rho\bigl(\mathcal{D}\mathcal{T}_{\mathrm{seq}}^{\pi}(K^\star)\bigr)
=
\rho\bigl(\mathcal{D}\mathcal{T}_{\mathrm{dec}}(K^\star)\bigr)^{2}.
\end{equation}
\end{corollary}

\begin{proof}
By \eqref{eq:app-GKi}, \(\mathcal{D}_K\mathcal{G}_i\) contains no term in
\(\Delta K_i\), and by \eqref{eq:app-EKi} its \(i\)-th block is
\(\mathcal{D}_{K_i}\mathcal{E}_i[\Delta K_i]
=-\Delta K_i^\top M_i-M_i^\top\Delta K_i\) with
\(M_i:=B_i^\top P_i^\star+D_i^\top P_i^\star C_{\mathrm{cl}}(K^\star)
-R_{ii}K_i^\star\). Since
\(C-\sum_{j\neq i}D_jK_j^\star=C_{\mathrm{cl}}(K^\star)+D_iK_i^\star\),
condition \eqref{eq:equilibrium-improve} yields \(M_i=0\), so that
\(\bigl[\mathcal{D}\mathcal U_i(K^\star)\bigr]_i=0\) by
\eqref{eq:DUi-operator} and \(D_\pi=0\) in \eqref{eq:Mpi-splitting}.

Write \(X_{ij}:=(\mathcal{D}\mathcal{T}_{\mathrm{dec}}^{\pi})_{ij}\). Then
\eqref{eq:operator-splitting} reduces to \((I-L_\pi)^{-1}R_\pi\), which is
block upper triangular with diagonal blocks \(0\) and \(X_{21}X_{12}\), whence
\(\rho(\mathcal{D}\mathcal{T}_{\mathrm{seq}}^{\pi}(K^\star))=\rho(X_{21}X_{12})\).
Moreover \((\mathcal{D}\mathcal{T}_{\mathrm{dec}}^{\pi})^{2}
=\operatorname{blkdiag}(X_{12}X_{21},X_{21}X_{12})\), and \(X_{12}X_{21}\) and
\(X_{21}X_{12}\) share their nonzero spectrum.
\end{proof}

For every two-player game, the two schemes therefore converge and diverge together at each stabilizing equilibrium, and where they converge, one sequential iteration contracts as much as two decoupled iterations.

\subsection{Effects of the Game Structure}
\label{sec:sdg_induced_convergence_effects}

In addition to the main structural results established in Theorem~\ref{thm:seq-no-worse} and Corollary~\ref{cor:two-player-squaring}, we now analyze how the stochasticity of the game affects local convergence. Specifically, we reveal how the control-dependent noise terms \(D_i\) enter the contraction bounds through the coefficients \(\gamma_{ij}\).

Let \(K^\star\in\mathcal K_{\mathrm{MSS}}\) be a stabilizing linear feedback Nash equilibrium in the sense of Definition~\ref{def:LFNE}, equivalently, a solution of \eqref{eq:equilibrium-eval}--\eqref{eq:equilibrium-improve}. For each \(i\in\mathcal I\), define
$
P_i^\star:=P_i(K^\star).
$
Moreover, let
$
(\mathcal{D}_P\mathcal{E}_i)^\star:=\mathcal{D}_P\mathcal{E}_i(K^\star)$, and
$
(\mathcal{D}_P\mathcal{G}_i)^\star:=\mathcal{D}_P\mathcal{G}_i(K^\star),
$
where \(\mathcal{D}_P\mathcal{E}_i\) and \(\mathcal{D}_P\mathcal{G}_i\) are the partial Fr\'echet derivatives introduced before Lemma~\ref{lem:reduced-playerwise}. For \(j\in\mathcal I\setminus\{i\}\), let
$
(\mathcal{D}_{K_j}\mathcal{E}_i)^\star:\mathbb R^{m_j\times n}\to\mathbb S^n
$
denote the restriction of \(\mathcal{D}_K\mathcal{E}_i(K^\star)\) to perturbations in the \(j\)-th player block, where \(\mathcal{D}_K\mathcal{E}_i\) is introduced before Lemma~\ref{lem:reduced-playerwise}. For each \(i\in\mathcal I\), define
\begin{equation}
\label{eq:gammaii}
\mu_i^\star:=\bigl\|\bigl((\mathcal{D}_P\mathcal{E}_i)^\star\bigr)^{-1}\bigr\|^{-1}, \,
\gamma_{ii}:=\bigl\|[\mathcal{D}\mathcal{T}_{\mathrm{inv}}]_{ii}\bigr\|,
\, i\in\mathcal I,
\end{equation}
where \(\mathcal{D}\mathcal{T}_{\mathrm{inv}}\) is given by \eqref{eq:Jinv-definition}, equivalently, by Proposition~\ref{prop:decoupled-jacobian},
$
\mathcal{D}\mathcal{T}_{\mathrm{inv}}=\mathcal{D}\mathcal{T}_{\mathrm{dec}}(K^\star).
$
For \(i\in\mathcal I\) and \(j\in\mathcal I\setminus\{i\}\), define
\begin{align}
\label{eq:gammaij}
\gamma_{ij}
&:=
\Bigl\|\bigl(R_{ii}+D_i^\top P_i^\star D_i\bigr)^{-1}D_i^\top P_i^\star D_j\Bigr\|
\nonumber\\
&\quad+
\frac{
\Bigl\|\bigl(R_{ii}+D_i^\top P_i^\star D_i\bigr)^{-1}\Bigr\|
\,\|(\mathcal{D}_P\mathcal{G}_i)^\star\|\,\|(\mathcal{D}_{K_j}\mathcal{E}_i)^\star\|
}{
\mu_i^\star
}.
\end{align}

The coefficients \(\gamma_{ij}\) in \eqref{eq:gammaij} reflect three structural drivers of local convergence. The first term captures cross-player control-dependent diffusion through \(D_i^\top P_i^\star D_j\), the second captures policy-evaluation sensitivity through \((\mu_i^\star)^{-1}\), and both are premultiplied by \(\bigl(R_{ii}+D_i^\top P_i^\star D_i\bigr)^{-1}\), which encodes how player \(i\)'s own control-dependent diffusion scales these effects. The following corollary uses \(\gamma_{ij}\) as a computable sufficient condition for local contraction. The criterion \eqref{eq:Jinv-weighted-comparison} is a weighted block-diagonal-dominance condition in the sense of \cite{feingold_varga_1962}.

\begin{corollary}[Explicit sufficient convergence criterion]
\label{cor:Jinv-block-comparison}
Let \(\pi\) be a permutation of \(\mathcal I\), and consider the transported
norm \(\|\cdot\|_{\pi,w,\infty}\) from \eqref{eq:transported-block-norm}. Assume
that there exist weights \(w_1,\dots,w_N>0\) and a constant \(q\in[0,1)\) such
that
\begin{equation}
\label{eq:Jinv-weighted-comparison}
\sum_{j\in\mathcal I} \frac{w_j}{w_i}\,\gamma_{\pi(i)\pi(j)} \le q,
\qquad i\in\mathcal I.
\end{equation}
Then
\begin{equation}
\label{eq:seq-contractive}
\|\mathcal{D}\mathcal{T}_{\mathrm{seq}}^{\pi}(K^\star)\|_{\pi,w,\infty}\le q<1.
\end{equation}
In particular,
\begin{equation}
\label{eq:seq-spectral}
\rho\!\left(\mathcal{D}\mathcal{T}_{\mathrm{seq}}^{\pi}(K^\star)\right)<1.
\end{equation}
Hence the sequential policy-iteration map \(\mathcal{T}_{\mathrm{seq}}^{\pi}\) is locally
contractive at \(K^\star\) in the weighted block-\(\ell_\infty\) norm
associated with the update order \(\pi\).
\end{corollary}

\begin{proof}
See Appendix~\ref{app:proof-explicit-contraction}.
\end{proof}

Notably, the influence of player \(i\)'s own diffusion matrix \(D_i\) on the criterion \eqref{eq:Jinv-weighted-comparison} is not monotone: \(D_i\) enters \(\gamma_{ij}\) both through the premultiplier \(\bigl(R_{ii}+D_i^\top P_i^\star D_i\bigr)^{-1}\) and through \(\mu_i^\star\), so that larger \(D_i\) may either tighten or relax the sufficient condition depending on the balance of these two effects.

The results established in this section are evaluated numerically in Section~\ref{sec:numerical_experiments}.
The analytic characterization of the Fr\'echet derivative of the sequential policy-iteration map developed in this section separates the effects of the update order from those of the game structure and provides explicit computable sufficient conditions for local contraction. It thereby completes Contribution~2 and resolves Problem~\ref{prob:convergence}.

\section{Homotopy Initialization}
\label{sec:homotopy_init}

This section develops a homotopy continuation method that produces a mean-square-stabilizing gain profile and thereby removes the requirement of a stabilizing gain profile for Algorithm~\ref{alg:cascaded-PI}, closing the initialization gap that arises from the restriction of the presented policy iteration method and its convergence analysis to \(\mathcal K_{\mathrm{MSS}}\). 
Mean-square stability is preserved throughout, so that every intermediate gain profile is itself admissible.
The procedure terminates after finitely many continuation steps, with an explicit bound on their number.
\subsection{Stability-Preserving Homotopy Path}
This subsection establishes the analytical foundation of the homotopy-based initialization by formulating an auxiliary single-controller problem and deriving the properties needed to obtain an MSS starting gain profile. The initialization procedure operates on this auxiliary problem and constructs a stabilizing gain profile by following a homotopy path from an artificially drift-shifted MSS system, initialized at $K = 0$, toward the target dynamics while successively removing the shift and preserving mean-square stability throughout. The resulting gain profile is admissible for Algorithm~\ref{alg:cascaded-PI}.

To construct a stabilizing initial gain profile for the \(N\)-player stochastic differential game
\eqref{eq:sde-N}--\eqref{eq:cost}, we temporarily disregard the game structure and treat the inputs of all players as the control of a single auxiliary controller. Collecting the input matrices as
\begin{equation}
\label{eq:aux-input-matrices}
B:=\begin{bmatrix}B_1&\cdots&B_N\end{bmatrix},
\qquad
D:=\begin{bmatrix}D_1&\cdots&D_N\end{bmatrix},
\end{equation}
with \(B,D\in\mathbb R^{n\times m}\) and \(m:=\sum_{i\in\mathcal I}m_i\), and identifying the joint control \(u=(u_1,\dots,u_N)\) with the stacked vector \((u_1^\top,\dots,u_N^\top)^\top\in\mathbb R^{m}\), the state equation \eqref{eq:sde-N} takes the single-input form
\begin{equation}
\label{eq:aux-sde}
dx(t)
=
\bigl(Ax(t)+Bu(t)\bigr)\,dt
+
\bigl(Cx(t)+Du(t)\bigr)\,dW(t).
\end{equation}
Under the corresponding identification of a gain profile \(K=(K_i)_{i\in\mathcal I}\in\mathcal K\) with the stacked matrix \(\begin{bmatrix}K_1^\top&\cdots&K_N^\top\end{bmatrix}^\top\in\mathbb R^{m\times n}\), the auxiliary control law \(u(t)=-Kx(t)\) coincides with the game feedback \(u_i(t)=-K_ix(t)\), \(i\in\mathcal I\), and
\begin{equation}
\label{eq:aux-closed-loop}
A-BK=A_{\mathrm{cl}}(K),
\qquad
C-DK=C_{\mathrm{cl}}(K).
\end{equation}
The auxiliary closed-loop system therefore coincides with the closed-loop game dynamics induced by \(K\).

The auxiliary controller minimizes the single cost functional
\begin{equation}
\label{eq:aux-cost}
J^{\mathrm{aux}}(u;x_0)
:=
\mathbb E\int_0^\infty
\Bigl(
x(t)^\top Q x(t)
+
u(t)^\top R\,u(t)
\Bigr)\,dt,
\end{equation}
with weights \(Q\in\mathbb S^n_{++}\) and \(R\in\mathbb S^m_{++}\), which may for instance be inherited from the game as \(Q=\sum_{i\in\mathcal I}Q_i\) and \(R=\operatorname{blkdiag}(R_{11},\dots,R_{NN})\).
Thus the initialization problem is to compute a stabilizing gain \(K\) for the auxiliary stochastic control problem \eqref{eq:aux-sde}--\eqref{eq:aux-cost}. By \eqref{eq:aux-closed-loop}, such a gain is precisely an element of \(\mathcal K_{\mathrm{MSS}}\).
The weights provide a degree of freedom in selecting where in \(\mathcal K_{\mathrm{MSS}}\) the initial profile is placed.

The auxiliary problem \eqref{eq:aux-sde}--\eqref{eq:aux-cost} is an instance of the game class \eqref{eq:sde-N}--\eqref{eq:cost} with a single player. The results of Section~\ref{sec:03_CSPI} therefore apply to it verbatim. In particular, its optimality conditions are obtained by specializing the coupled equilibrium relations \eqref{eq:equilibrium-eval}--\eqref{eq:equilibrium-improve} to the auxiliary dynamics \eqref{eq:aux-sde} with the single cost \eqref{eq:aux-cost}.

Assumption~\ref{ass:MSS} is existential and does not yield an explicit stabilizing gain.
To initialize the homotopy at the trivial gain \(K=0\), we therefore introduce a uniform negative shift in the drift of the auxiliary system. The shift parameter \(\beta\) is chosen so that the shifted auxiliary system is asymptotically mean-square stable already at \(K=0\).

\begin{lemma}[Sufficient drift shift for mean-square stability]
\label{lem:beta_shift_mss}
Fix \(\varepsilon_\beta>0\) and define
\begin{equation}
\label{eq:beta_def}
\beta
:=
\max\Bigl\{
\tfrac12\lambda_{\max}(A^\top+A+C^\top C),\,0
\Bigr\}
+\varepsilon_\beta.
\end{equation}
Then the shifted auxiliary system
\begin{equation}
\label{eq:shifted_aux_zero_gain}
dx(t)
=
(A-\beta I)x(t)\,dt
+
Cx(t)\,dW(t),
\qquad t\ge 0,
\end{equation}
which corresponds to the auxiliary dynamics \eqref{eq:aux-sde} at the trivial gain \(K=0\), is asymptotically mean-square stable.
\end{lemma}

\begin{proof}
By \eqref{eq:Lyap-MSS}, it suffices to verify
\begin{equation}
(A-\beta I)^\top X + X(A-\beta I) + C^\top X C \prec 0
\end{equation}

for some \(X=X^\top\succ0\). Choosing \(X=I\) yields
\begin{equation}
A^\top + A + C^\top C - 2\beta I \prec 0
\end{equation}

by \eqref{eq:beta_def}. Hence \eqref{eq:Lyap-MSS} holds, and
\eqref{eq:shifted_aux_zero_gain} is asymptotically mean-square stable.
\end{proof}

Introduce the homotopy level \(\sigma\geq 0\) via the shifted drift \(A-\beta I+\sigma I\), with \(\sigma=0\) corresponding to the explicitly stabilized starting point and \(\sigma=\beta\) to the unshifted auxiliary dynamics.

\begin{lemma}[MSS preservation under the auxiliary PI step]
\label{lem:aux_mss_preservation}
Fix \(\sigma\in[0,\beta]\) and replace the drift matrix \(A\) in the auxiliary
system \eqref{eq:aux-sde} by \(A-\beta I+\sigma I\).
Let \(K\in\mathcal K\) be such that the corresponding closed-loop pair
\begin{equation}
\label{eq:shifted-aux-cl}
\bigl(A-\beta I+\sigma I-BK,\; C-DK\bigr)
\end{equation}
is asymptotically mean-square stable. Let \(P\) denote the unique solution of
the associated policy-evaluation equation
\begin{align}
\label{eq:shifted-aux-eval}
0
&=
\bigl(A-\beta I+\sigma I-BK\bigr)^\top P
+
P\bigl(A-\beta I+\sigma I-BK\bigr)
\notag\\
&\quad+
\bigl(C-DK\bigr)^\top P\bigl(C-DK\bigr)
+
Q+K^\top R\,K,
\end{align}
and define
\begin{equation}
\label{eq:aux-improvement}
\widehat K
=
\bigl(R+D^\top P D\bigr)^{-1}
\bigl(B^\top P+D^\top P C\bigr).
\end{equation}
Then the updated closed-loop pair
\begin{equation}
\label{eq:shifted-aux-cl-updated}
\bigl(A-\beta I+\sigma I-B\widehat K,\; C-D\widehat K\bigr)
\end{equation}
is asymptotically mean-square stable.
\end{lemma}

\begin{proof}
Since \eqref{eq:shifted-aux-cl} is asymptotically mean-square stable,
Lemma~\ref{lem:eval-Pi} yields existence and uniqueness of the solution
\(P\) to \eqref{eq:shifted-aux-eval}. The claim is then the specialization of
Theorem~\ref{thm:mss-preservation-cascaded} to a single controller with input
matrices \eqref{eq:aux-input-matrices}, weights \(Q\succ0\) and \(R\succ0\), and
drift \(A-\beta I+\sigma I\).
\end{proof}

We next derive a sufficient condition on the homotopy increment for preservation of mean-square stability.

\begin{lemma}[Admissible homotopy step size]
\label{lem:homotopy_stepsize}
Fix \(\sigma_k\in[0,\beta]\), and let \(K^{k+1}\in\mathcal K\) be such that
\begin{equation}
\label{eq:hatAk-hatCk}
\widehat A_{k+1}:=A-\beta I+\sigma_k I-BK^{k+1},
\quad
\widehat C_{k+1}:=C-DK^{k+1}
\end{equation}
defines an asymptotically mean-square stable closed-loop pair. Let
\(\widehat P^k\succ0\) denote the unique solution of
\eqref{eq:shifted-aux-eval} with \(\sigma=\sigma_k\) and
\(K=K^{k+1}\), and define
\begin{equation}
\label{eq:gammak-def}
\nu_k
:=
\lambda_{\min}\!\Bigl(
(\widehat P^k)^{-1/2}
\bigl(Q+(K^{k+1})^\top R\,K^{k+1}\bigr)
(\widehat P^k)^{-1/2}
\Bigr).
\end{equation}
Then \(\nu_k>0\), and every
$
\alpha_{k+1}\in\Bigl(0,\frac{\nu_k}{2}\Bigr)
$
has the property that
\begin{equation}
\label{eq:next-shifted-pair}
\bigl(\widehat A_{k+1}+\alpha_{k+1}I,\widehat C_{k+1}\bigr)
\end{equation}
is asymptotically mean-square stable.
\end{lemma}

\begin{proof}
Since \(Q\succ0\) and \(R\succ0\), the matrix
\begin{equation}
Q+(K^{k+1})^\top R\,K^{k+1}
\end{equation}

is positive definite, and hence \(\nu_k>0\). Moreover, by
\eqref{eq:shifted-aux-eval},
\begin{align}
\widehat A_{k+1}^\top \widehat P^k
&+
\widehat P^k \widehat A_{k+1}
+
\widehat C_{k+1}^\top \widehat P^k \widehat C_{k+1}
\notag\\
&=
-
\bigl(Q+(K^{k+1})^\top R\,K^{k+1}\bigr).
\end{align}
Therefore, for any \(\alpha>0\),
\begin{align}
(\widehat A_{k+1}+\alpha I)^\top \widehat P^k
&+
\widehat P^k(\widehat A_{k+1}+\alpha I)
+
\widehat C_{k+1}^\top \widehat P^k \widehat C_{k+1} \notag
\\
&=
-
\bigl(Q+(K^{k+1})^\top R\,K^{k+1}\bigr)
+
2\alpha \widehat P^k.
\end{align}
If \(\alpha\in(0,\nu_k/2)\), then by \eqref{eq:gammak-def},
\begin{equation}
(\widehat P^k)^{-1/2}
\bigl(Q+(K^{k+1})^\top R\,K^{k+1}\bigr)
(\widehat P^k)^{-1/2}
-
2\alpha I
\succ0,
\end{equation}

and hence
$
-
\bigl(Q+(K^{k+1})^\top R\,K^{k+1}\bigr)
+
2\alpha \widehat P^k
\prec0.
$
Thus the mean-square Lyapunov inequality \eqref{eq:Lyap-MSS} holds for
\eqref{eq:next-shifted-pair} with \(\widehat P^k\), which proves
asymptotic mean-square stability.
\end{proof}

Whereas Lemma~\ref{lem:homotopy_stepsize} bounds each homotopy step so as to preserve mean-square stability along the path, the last step need not meet \(\beta\) exactly and may yield \(\sigma_k>\beta\). The following lemma shows that the resulting gain profile satisfies \(K\in\mathcal K_{\mathrm{MSS}}\) nonetheless, so that the continuation can be stopped as soon as \(\sigma_k\ge\beta\).

\begin{lemma}[Safe overshoot]
\label{lem:overshoot_safe}
Let \(\beta\) be given by \eqref{eq:beta_def}, let \(\sigma\ge\beta\), and let \(K\in\mathcal K\) be such that
the shifted closed-loop pair \eqref{eq:shifted-aux-cl} is
asymptotically mean-square stable. Then \(K\in\mathcal K_{\mathrm{MSS}}\).
\end{lemma}

\begin{proof}
Set \(\delta:=\sigma-\beta\ge0\). Then, by \eqref{eq:aux-closed-loop}, the shifted closed-loop pair
\eqref{eq:shifted-aux-cl} is
\begin{equation}
\bigl(A_{\mathrm{cl}}(K)+\delta I,\; C_{\mathrm{cl}}(K)\bigr).
\end{equation}

By \eqref{eq:Lyap-MSS}, there exists \(P=P^\top\succ0\) such that
\begin{align}
\bigl(A_{\mathrm{cl}}(K)+\delta I\bigr)^\top P
&+
P\bigl(A_{\mathrm{cl}}(K)+\delta I\bigr)
\notag\\
&\quad+
\bigl(C_{\mathrm{cl}}(K)\bigr)^\top
P\,C_{\mathrm{cl}}(K)
\prec0.
\end{align}

Subtracting \(2\delta P\succeq0\) yields
\begin{equation}
\bigl(A_{\mathrm{cl}}(K)\bigr)^\top P
+
P A_{\mathrm{cl}}(K)
+
\bigl(C_{\mathrm{cl}}(K)\bigr)^\top
P\,C_{\mathrm{cl}}(K)
\prec0.
\end{equation}
Hence \eqref{eq:Lyap-MSS} holds for the closed-loop pair induced by \(K\) in the
game \eqref{eq:sde-N}, that is, \(K\in\mathcal K_{\mathrm{MSS}}\).
\end{proof}

\begin{algorithm}[t]
\caption{Homotopy Initialization for Stochastic $N$-Player LQ Differential Games}
\label{alg:homotopy_init}
\begin{algorithmic}
\REQUIRE system matrices \(A,C\) and \(B,D\) from \eqref{eq:aux-input-matrices}, weights \(Q\succ0\), \(R\succ0\), shift margin \(\varepsilon_\beta>0\), step-size factor \(\eta\in(0,1)\)
\STATE Set \(k\gets 0\), \(\alpha_0\gets 0\), \(K^0\gets 0\)
\STATE Compute \(\beta\) from \eqref{eq:beta_def}
\REPEAT
    \STATE Set
    $
    \sigma_k:=\sum_{j=0}^k \alpha_j
    $
    \STATE \textbf{Policy iteration:} Starting from \(K^k\), iterate at level \(\sigma_k\) the policy-evaluation step
    \[
    \begin{aligned}
    0={}&\bigl(A-\beta I+\sigma_k I-BK\bigr)^\top P \\
    &\quad +P\bigl(A-\beta I+\sigma_k I-BK\bigr) \\
    &\quad +\bigl(C-DK\bigr)^\top P\bigl(C-DK\bigr) \\
    &\quad +Q+K^\top R\,K
    \end{aligned}
    \]
    and the policy-improvement update
    \[
    K\;\gets\;\bigl(R+D^\top P D\bigr)^{-1}
    \bigl(B^\top P+D^\top P C\bigr)
    \]
    until convergence, yielding the stabilizing solution \(\widehat P^k\succ0\) and the optimal gain \(K^{k+1}\)
    \STATE \textbf{Margin:} Compute \(\nu_k\) from \eqref{eq:gammak-def}
    \STATE Choose
    $
    \alpha_{k+1}:=\eta\,\frac{\nu_k}{2}
    $
    \STATE Set \(k\gets k+1\)
\UNTIL \(\sigma_k\ge\beta\)
\STATE Set \(K^{\mathrm{init}}:=K^k\) and partition it row-wise into the player blocks \(K_i^{\mathrm{init}}\in\mathbb R^{m_i\times n}\), \(i\in\mathcal I\)
\RETURN \(K^{\mathrm{init}}\)
\end{algorithmic}
\end{algorithm}

\subsection{Homotopy Initialization Algorithm}

Starting from the explicitly stabilized auxiliary system from Lemma~\ref{lem:beta_shift_mss}, the homotopy initialization proceeds by repeated policy-evaluation and policy-improvement steps, combined with a stability-preserving update of the homotopy level. The procedure is summarized in Algorithm~\ref{alg:homotopy_init}.
\begin{theorem}[Well-posedness of Algorithm~\ref{alg:homotopy_init}]
\label{thm:homotopy_wellposed}
Let \(\varepsilon_\beta>0\), \(\eta\in(0,1)\), let \(Q\succ0\) and \(R\succ0\),
and let \(\beta\) be given by \eqref{eq:beta_def}. For
each \(k\ge0\), define
\begin{equation}
\label{eq:theorem-sigmak}
\sigma_k:=\sum_{j=0}^k \alpha_j
\end{equation}
and the shifted closed-loop pair
\begin{equation}
\label{eq:theorem-AkCk}
A_k:=A-\beta I+\sigma_k I-BK^k,
\qquad
C_k:=C-DK^k.
\end{equation}
If, for some \(k\ge0\), the pair \((A_k,C_k)\) is asymptotically mean-square
stable, then the next iteration of Algorithm~\ref{alg:homotopy_init} is well
defined and produces a step size \(\alpha_{k+1}>0\) such that
\begin{equation}
\label{eq:theorem-next-pair}
A_{k+1}:=A-\beta I+\sigma_{k+1}I-BK^{k+1},
\quad
C_{k+1}:=C-DK^{k+1}
\end{equation}
is again asymptotically mean-square stable. In particular, \((A_k,C_k)\) is
asymptotically mean-square stable for all \(k\ge0\).

Moreover, under Assumption~\ref{ass:MSS}, \(\sigma_k\ge\beta\) is reached after
finitely many iterations, and the corresponding unshifted closed-loop pair
\begin{equation}
\label{eq:theorem-unshifted-pair}
\bigl(A-BK^k,\; C-DK^k\bigr)
\end{equation}
is then asymptotically mean-square stable.
\end{theorem}

\begin{proof}
We argue by induction on \(k\). For \(k=0\), Algorithm~\ref{alg:homotopy_init}
sets \(K^0=0\) and \(\alpha_0=0\), so that \(\sigma_0=0\) and
\((A_0,C_0)=(A-\beta I,\,C)\). The claim therefore follows from
Lemma~\ref{lem:beta_shift_mss}.

Now assume that \((A_k,C_k)\) is asymptotically mean-square stable, so that \(K^k\) stabilizes the auxiliary system at level \(\sigma_k\). Policy iteration at level \(\sigma_k\), started from \(K^k\), is then well defined, since by Lemma~\ref{lem:eval-Pi} each policy-evaluation equation \eqref{eq:shifted-aux-eval} (with \(\sigma=\sigma_k\)) admits a unique solution \(P\succ0\), each improvement is well defined since
\begin{equation}
\label{eq:theorem-improvement-wellposed}
R+D^\top P D\succ0,
\end{equation}
and by Lemma~\ref{lem:aux_mss_preservation} mean-square stability is preserved at every step. The iteration converges to the optimal stabilizing gain \(K^{k+1}\), with value \(\widehat P^k\succ0\)~\cite{damm2004}, so that the pair
\begin{equation}
\label{eq:theorem-hatAkhatCk}
\widehat A_{k+1}:=A-\beta I+\sigma_k I-BK^{k+1},
\quad
\widehat C_{k+1}:=C-DK^{k+1}
\end{equation}
is asymptotically mean-square stable.
Applying Lemma~\ref{lem:homotopy_stepsize} yields \(\nu_k>0\). Since
Algorithm~\ref{alg:homotopy_init} chooses
\begin{equation}
\label{eq:theorem-alpha-choice}
\alpha_{k+1}=\eta\,\frac{\nu_k}{2},
\end{equation}
one has \(\alpha_{k+1}\in(0,\nu_k/2)\). Therefore,
\begin{equation}
\label{eq:theorem-shifted-next}
(\widehat A_{k+1}+\alpha_{k+1}I,\widehat C_{k+1})
=
(A_{k+1},C_{k+1})
\end{equation}
is asymptotically mean-square stable. This proves that the next iteration is
well defined and preserves mean-square stability.

It remains to establish finite termination. Let \(\bar K\in\mathcal K_{\mathrm{MSS}}\), which is nonempty by Assumption~\ref{ass:MSS}. By \eqref{eq:aux-closed-loop}, \(\bar K\) stabilizes the unshifted auxiliary system. Since a larger drift shift only improves stability, \(\bar K\) is admissible at every level \(\sigma\in[0,\beta]\). Let \(P_\beta\succ0\) denote the solution of \eqref{eq:shifted-aux-eval} for \(\bar K\) at \(\sigma=\beta\).
Let \(P_{\sigma_k}\succ0\) denote the solution of \eqref{eq:shifted-aux-eval} for \(\bar K\) at \(\sigma=\sigma_k\). Subtracting the two evaluation equations shows that \(\Delta:=P_\beta-P_{\sigma_k}\) satisfies
\begin{equation}
\label{eq:theorem-value-monotone}
A_{\mathrm{cl}}(\bar K)^\top\Delta+\Delta A_{\mathrm{cl}}(\bar K)+C_{\mathrm{cl}}(\bar K)^\top\Delta\,C_{\mathrm{cl}}(\bar K)=-2(\beta-\sigma_k)P_{\sigma_k},
\end{equation}
whose right-hand side is negative semidefinite because \(\sigma_k\le\beta\). Since \(\bar K\in\mathcal K_{\mathrm{MSS}}\), the generalized Lyapunov equation \eqref{eq:theorem-value-monotone} is uniquely solvable, and a negative semidefinite right-hand side yields \(\Delta\succeq0\) by the argument used in the proof of Lemma~\ref{lem:eval-Pi} with the strict inequalities replaced by their nonstrict counterparts. Hence \(P_{\sigma_k}\preceq P_\beta\). By optimality of \(\widehat P^k\) at level \(\sigma_k\), it follows that \(\widehat P^k\preceq P_{\sigma_k}\preceq P_\beta\), and in particular
\begin{equation}
\lambda_{\max}(\widehat P^k)\le\lambda_{\max}(P_\beta)
\end{equation}
is bounded by a constant independent of \(k\). By \eqref{eq:gammak-def}, dropping the nonnegative term \((K^{k+1})^\top R\,K^{k+1}\) and applying the standard bound \(\lambda_{\min}(X^{-1/2}YX^{-1/2})\ge\lambda_{\min}(Y)/\lambda_{\max}(X)\) for \(X,Y\succ0\) (see, e.g.,~\cite[Thm.~4.2.2]{HornJohnson}),
\begin{align}
\label{eq:theorem-gamma-lower}
\nu_k
\;&\ge\;
\lambda_{\min}\!\bigl((\widehat P^k)^{-1/2}Q(\widehat P^k)^{-1/2}\bigr) \nonumber \\
\;&\ge\;
\frac{\lambda_{\min}(Q)}{\lambda_{\max}(\widehat P^k)}
\;\ge\;
\frac{\lambda_{\min}(Q)}{\lambda_{\max}(P_\beta)}
\;>\;0.
\end{align}
Hence the homotopy level increases at every iteration by at least
\begin{equation}
\label{eq:theorem-alpha-lower}
\alpha_{k+1}=\eta\,\frac{\nu_k}{2}\;\ge\;\underline\alpha:=\frac{\eta\,\lambda_{\min}(Q)}{2\lambda_{\max}(P_\beta)}\;>\;0,
\end{equation}
so that \(\sigma_k\ge k\,\underline\alpha\) and therefore \(\sigma_k\ge\beta\) after at most
\begin{equation}
\label{eq:theorem-step-bound}
\Bigl\lceil\frac{\beta}{\underline\alpha}\Bigr\rceil
=
\Bigl\lceil\frac{2\beta\,\lambda_{\max}(P_\beta)}{\eta\,\lambda_{\min}(Q)}\Bigr\rceil
<\infty
\end{equation}
iterations. Once \(\sigma_k\ge\beta\), Lemma~\ref{lem:overshoot_safe} applied to
\((A_k,C_k)\) yields asymptotic mean-square stability of the unshifted
pair \eqref{eq:theorem-unshifted-pair}.
\end{proof}

Assumption~\ref{ass:MSS} ensures that the target problem at the homotopy level \(\sigma=\beta\), i.e., the unshifted auxiliary system, admits a stabilizing gain. Theorem~\ref{thm:homotopy_wellposed} shows that Algorithm~\ref{alg:homotopy_init} preserves asymptotic mean-square stability along the homotopy path and that every iterate with \(\sigma_k\ge\beta\) already yields a stabilizing gain for the unshifted auxiliary system.

Let \(k\) denote the terminal iteration index of Algorithm~\ref{alg:homotopy_init}, so that its output satisfies \(K^{\mathrm{init}}=K^k\) with \(\sigma_k\ge\beta\). By Theorem~\ref{thm:homotopy_wellposed}, the unshifted pair \(\bigl(A-BK^{\mathrm{init}},\,C-DK^{\mathrm{init}}\bigr)\) is asymptotically mean-square stable. By \eqref{eq:aux-closed-loop}, this pair coincides with \(\bigl(A_{\mathrm{cl}}(K^{\mathrm{init}}),\,C_{\mathrm{cl}}(K^{\mathrm{init}})\bigr)\), so that the closed-loop game dynamics induced by the gain profile \(K^{\mathrm{init}}=(K_i^{\mathrm{init}})_{i\in\mathcal I}\) is asymptotically mean-square stable, that is, \(K^{\mathrm{init}}\in\mathcal K_{\mathrm{MSS}}\).

Algorithm~\ref{alg:homotopy_init} thus constructs a stabilizing initial gain profile \(K^{\mathrm{init}}\in\mathcal K_{\mathrm{MSS}}\) for every mean-square stabilizable stochastic \(N\)-player LQ differential game, addressing the initialization requirement of Algorithm~\ref{alg:cascaded-PI}, thereby completing Contribution~3 and resolving Problem~\ref{prob:init}.

\section{Numerical Examples}
\label{sec:numerical_experiments}

We examine the local convergence behavior of sequential policy iteration on a stochastic LQ differential game~\eqref{eq:sde-N}--\eqref{eq:cost} with \(N=3\), the smallest number of players at which the update order can matter, and \(n=2\), \(m_i=2\) for all \(i\in\mathcal I\), varying only the diffusion matrices.
With increasing diffusion intensity, the update order can determine whether local contraction is lost, so that its effect is not merely quantitative.
This dependence itself does not require diffusion, as the deterministic example in Lemma~\ref{lem:spectral-order-dependence} shows, and here the diffusion level determines its severity.
The game is constructed so that all these effects become visible in a single instance.
We consider a nominal game specified by the matrices
\[
A =
\begin{bmatrix}
-0.781 & -1.161 \\
-1.376 & -0.926
\end{bmatrix},
\quad
C =
\begin{bmatrix}
0.211 & 0.563 \\
0.524 & -0.408
\end{bmatrix},
\]
\[
B_1 =
\begin{bmatrix}
-1.325 & -0.429 \\
-2.411 & -1.187
\end{bmatrix},
\quad
B_2 =
\begin{bmatrix}
-2.106 & 3.394 \\
-0.183 & 2.958
\end{bmatrix},
\]
\[
B_3 =
\begin{bmatrix}
0.256 & 0.953 \\
-0.060 & 0.068
\end{bmatrix},
\quad
D_1 =
\begin{bmatrix}
-0.030 & -0.005 \\
-0.040 & 0.063
\end{bmatrix},
\]
\[
D_2 =
\begin{bmatrix}
0.237 & 0.290 \\
0.467 & 0.069
\end{bmatrix},
\quad
D_3 =
\begin{bmatrix}
-0.077 & -0.025 \\
0.090 & 0.075
\end{bmatrix}.
\]

The weighting matrices are \(Q_i=q_i I_2\) and \(R_{ij}=r_{ij} I_2\) with
\[
q=\begin{bmatrix}0.214 & 36.940 & 0.184\end{bmatrix},
\]
\[
[r_{ij}]=\begin{bmatrix}
2.026 & 0.003 & 2.575\\
0.002 & 2.948 & 0.007\\
0.168 & 2.308 & 41.851
\end{bmatrix}.
\]

For \(s\ge 0\), we vary only the diffusion terms according to
\[
C(s)=s\,C,
\qquad
D_i(s)=s\,D_i,
\qquad i\in\mathcal I,
\]
while keeping \(A\), \(B_i\), \(Q_i\), and \(R_{ij}\) fixed. Thus
\(s\) parametrizes the diffusion intensity. The case \(s=1\) is the
nominal stochastic game.
First, as a preliminary analysis Algorithm~\ref{alg:homotopy_init} was applied to the nominal game, for which $K=0$ is not mean-square stabilizing. With $Q=I$, $R=\operatorname{blkdiag}(R_{ii})$, $\eta=0.5$, and $\varepsilon_\beta=10^{-3}$, it returned a stabilizing profile after two continuation levels and $16$ policy-iteration steps, within the bound $\lceil 2\beta\lambda_{\max}(P_\beta)/(\eta\lambda_{\min}(Q))\rceil=5$ of Theorem~\ref{thm:homotopy_wellposed}. Algorithm~\ref{alg:cascaded-PI} then converged from this profile in nine iterations.

Second, as the main numerical example we compare the two update orders
$
\pi^{(1)}:=(2,1,3)$, and
$
\pi^{(2)}:=(1,2,3).
$
By Lemma~\ref{lem:spectral-order-dependence}, cyclic shifts of a given
order induce the same spectrum of the local Fr\'echet derivative
\(\mathcal{D}\mathcal{T}_{\mathrm{seq}}^{\pi}(K^\star)\). Hence, for \(N=3\), the six
permutations split into two cyclic equivalence classes, represented by
\(\pi^{(1)}\) and \(\pi^{(2)}\).
For each value of \(s\), we compute a stabilizing equilibrium
\(K^\star(s)\in\mathcal K_{\mathrm{MSS}}\) by
Algorithm~\ref{alg:cascaded-PI} with update order \(\pi^{(1)}\). 
Along the parameter sweep, the converged equilibrium at the previous value of $s$ is used as the initial gain profile for the next value, so that a single equilibrium branch is followed, with Algorithm~\ref{alg:homotopy_init} supplying the profile at the first value. 
At each computed equilibrium  \(K^\star(s)\), we evaluate the local Fr\'echet derivatives
$
\mathcal{D}\mathcal{T}_{\mathrm{seq}}^{\pi^{(1)}}(K^\star(s))$,
$
\mathcal{D}\mathcal{T}_{\mathrm{seq}}^{\pi^{(2)}}(K^\star(s))$, and
$
\mathcal{D}\mathcal{T}_{\mathrm{dec}}(K^\star(s)),
$
together with their spectral radii and the weighted
block-\(\ell_\infty\) quantities from
Theorem~\ref{thm:seq-no-worse}. For the weighted block-\(\ell_\infty\) quantities in
\eqref{eq:block-row-operator-norm}, the weights \(w\) are obtained by
normalizing a dominant eigenvector of the nonnegative matrix with entries
\(\|(\mathcal{D}\mathcal{T}_{\mathrm{dec}}^{\pi})_{ij}\|\).
All runs are terminated once
$
\|K^{k+1}-K^k\| \le 10^{-8},
$
or once \(500\) iterations are reached.

\subsection{Convergence Rates and Order Dependence}

The quantity evaluated first is the local spectral radius at \(K^\star(s)\).
By Lemma~\ref{lem:local-convergence}, \(\rho<1\) yields local contraction, whereas \(\rho>1\) renders the equilibrium locally repelling \cite{OrtegaRheinboldt}.
Figure~\ref{plt:rho} summarizes the local spectral radii
\(\rho_{\mathrm{seq}}^{\pi}(s):=\rho\!\left(\mathcal{D}\mathcal{T}_{\mathrm{seq}}^{\pi}(K^\star(s))\right)\)
and
\(\rho_{\mathrm{dec}}(s):=\rho\!\left(\mathcal{D}\mathcal{T}_{\mathrm{dec}}(K^\star(s))\right)\). All three quantities
increase with the diffusion level \(s\), showing that stronger diffusion
systematically weakens local contraction. At \(s=0\), all three update maps are
strongly contractive
(\(\rho_{\mathrm{seq}}^{\pi^{(1)}}(0)\approx 1.60\times10^{-2}\),
\(\rho_{\mathrm{seq}}^{\pi^{(2)}}(0)\approx 3.97\times10^{-3}\),
\(\rho_{\mathrm{dec}}(0)\approx 7.43\times10^{-2}\)),
and the same hierarchy is still visible at the nominal game
\(s=1\). As \(s\) increases further, the three spectral radii grow at markedly different rates, so that the local behavior of the three updates separates increasingly.

A central observation is that the locally preferred sequential update order is
not constant along the diffusion sweep. For small \(s\), the order \(\pi^{(2)}\) yields the smaller local spectral radius, but the two orders cross near \(s\approx2.2\), beyond which \(\pi^{(1)}\) becomes preferable.

Most importantly, the loss of local contraction is itself order-dependent. The
decoupled update crosses the threshold \(\rho=1\) first, near \(s\approx4.6\), and the sequential order \(\pi^{(2)}\) only later, near \(s\approx5.7\), whereas the order \(\pi^{(1)}\) remains locally contractive throughout the entire tested range.
Running all three schemes at \(s=6\) from a common admissible gain profile near
\(K^\star\) confirms that this spectral prediction governs the actual iteration.
In Fig.~\ref{fig:order-residuals}, the order \(\pi^{(1)}\) converges at the rate
predicted by Lemma~\ref{lem:local-convergence} and reaches machine precision
after roughly \(30\) iterations, whereas \(\pi^{(2)}\) is repelled from
\(K^\star\) at its own predicted rate. The decoupled scheme leaves
\(\mathcal K_{\mathrm{MSS}}\) after \(33\) iterations, so that its value
equations cease to be solvable and the iteration is no longer defined, while
both sequential orders remain admissible throughout, as guaranteed by
Theorem~\ref{thm:mss-preservation-cascaded}. 
The update order therefore governs not only the local rate, but whether local convergence is retained at all under increasing diffusion.
Sequential schemes have been used repeatedly \cite{nortmann_monti_sassano_mylvaganam2024_tac_iterative_datadriven_lq_games,chen_chen_lewis_xie_mcpi_nash_dg,chen_chen_lewis_2026_automatica_mcpi_nonlinear}, where the effect of the update order has remained unexplained.  To the best of our knowledge, it has not been identified as a factor that can decide whether an equilibrium is locally attracting at all in policy iteration for differential games.  
Its absence from the literature is consistent with Remark~\ref{rem:two-player-spectral-invariance} and Lemma~\ref{lem:spectral-order-dependence}, which confine the effect to at least three players and to orders from distinct cyclic classes.

\begin{figure}
    \centering
    \begin{tikzpicture}
    \begin{axis}[
        tacaxis,
        xlabel={Diffusion scaling $s$},
        ylabel={Spectral radius}
    ]
      \addplot[decstyle] table [x=s, y=rho_dec, col sep=comma] {content/plots/01_gamma_plot.csv};
      \addlegendentry{$\rho_{\mathrm{dec}}(s)$}
    
      \addplot[seqBstyle] table [x=s, y=rho_seqB, col sep=comma] {content/plots/01_gamma_plot.csv};
      \addlegendentry{$\rho_{\mathrm{seq}}^{\pi^{(2)}}(s)$}

      \addplot[seqAstyle] table [x=s, y=rho_seqA, col sep=comma] {content/plots/01_gamma_plot.csv};
      \addlegendentry{$\rho_{\mathrm{seq}}^{\pi^{(1)}}(s)$}
    
      \addplot[thresholdstyle] coordinates {(0,1) (6,1)};
    \end{axis}
    \end{tikzpicture}
    \caption{Local spectral radii
    $\rho_{\mathrm{seq}}^{\pi}(s):=\rho(\mathcal{D}\mathcal{T}_{\mathrm{seq}}^{\pi}(K^\star(s)))$
    and
    $\rho_{\mathrm{dec}}(s):=\rho(\mathcal{D}\mathcal{T}_{\mathrm{dec}}(K^\star(s)))$
    along the diffusion sweep.}
    \label{plt:rho}
\end{figure}
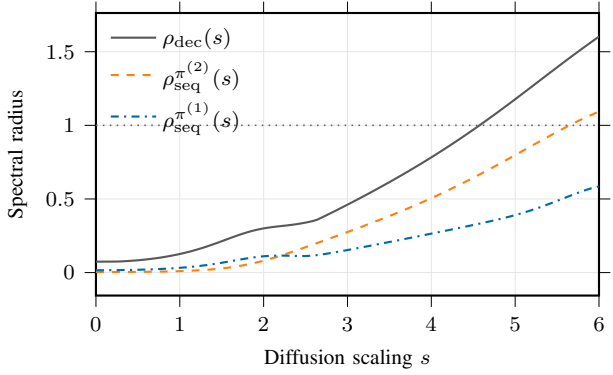

\subsection{Comparison of Sequential and Decoupled Dynamics}

Spectral radii yield sharp local rates but no guaranteed comparison between the
two schemes. Theorem~\ref{thm:seq-no-worse} supplies such a guarantee in a
transported weighted norm, which we evaluate along the same sweep.
In Fig.~\ref{fig:gamma_plot}, the comparison quantity for the decoupled core
increases monotonically with \(s\) and crosses the threshold \(1\) near
\(s\approx 2.76\), whereas both sequential orders remain strictly below it
throughout the tested range.
The plotted values are the exact block norms, not the explicit bounds of
Corollary~\ref{cor:Jinv-block-comparison}. 
The mechanism identified there
nonetheless accounts for the behavior, since increasing diffusion enlarges the
coupling coefficients \(\gamma_{ij}\) and reduces the contraction margin.
Whenever the decoupled norm does not exceed one, the sequential local
contraction is at least as strong as the decoupled one in this norm, and the
ordering seen in the spectral radii is consistent with the guarantee.
The condition of Theorem~\ref{thm:seq-no-worse} thus delimits the equilibria at
which the advantage of the sequential update structure is not incidental but
guaranteed.

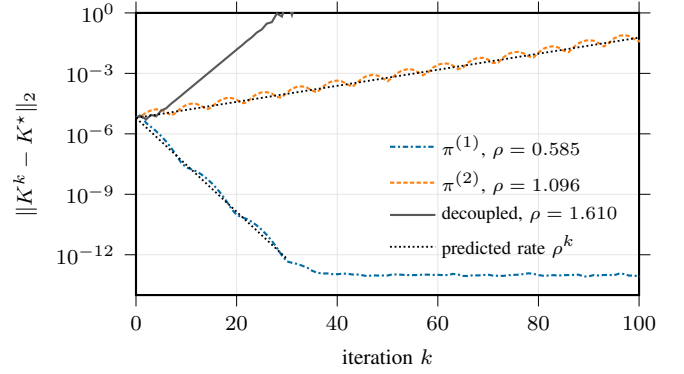
\begin{figure}[t]
    \centering
    \begin{tikzpicture}
    \begin{axis}[
        tacaxis,
        ymode=log,
        xmin=0, xmax=100,
        ymin=1e-14, ymax=1e0,
        ytick={1e-12,1e-9,1e-6,1e-3,1e0},
        unbounded coords=jump,
        xlabel={iteration $k$},
        ylabel={$\|K^{k}-K^{\star}\|_2$},
        legend style={
          at={(0.975,0.6)},
          anchor=north east,
          font=\scriptsize,
          draw=none,
          fill=none,
          row sep=0.5pt,
          inner sep=1pt
        },
    ]
      \addplot[seqAdensestyle] table [x=k, y=dist2_primary_seqA, col sep=comma] {content/plots/order_dependence_residuals_plot_data.csv};
      \addlegendentry{$\pi^{(1)}$, $\rho=0.585$}

      \addplot[seqBdensestyle] table [x=k, y=dist2_primary_seqB, col sep=comma] {content/plots/order_dependence_residuals_plot_data.csv};
      \addlegendentry{$\pi^{(2)}$, $\rho=1.096$}

      \addplot[decstyle] table [x=k, y=dist2_primary_dec, col sep=comma] {content/plots/order_dependence_residuals_plot_data.csv};
      \addlegendentry{decoupled, $\rho=1.610$}

      \addplot[predstyle] table [x=k, y=ref_seqA, col sep=comma] {content/plots/order_dependence_residuals_plot_data.csv};
      \addlegendentry{predicted rate $\rho^{k}$}

      \addplot[predstyle, forget plot] table [x=k, y=ref_seqB, col sep=comma] {content/plots/order_dependence_residuals_plot_data.csv};
    \end{axis}
    \end{tikzpicture}
    \caption{Distance to the equilibrium \(K^\star\) along the iteration at
    \(s=6\), for the two update orders of Fig.~\ref{plt:rho} and the decoupled
    scheme. Dotted lines show the rates \(\rho^{k}\) predicted by
    Lemma~\ref{lem:local-convergence}. The decoupled curve ends where its
    iterates leave \(\mathcal K_{\mathrm{MSS}}\).}
    \label{fig:order-residuals}
\end{figure}

\begin{figure}[t]
    \centering
    \begin{tikzpicture}
    \begin{axis}[
        tacaxis,
        xmin=0,
        xmax=2.6,
        xlabel={Diffusion scaling $s$},
        ylabel={Weighted block-norm}
    ]
      \addplot[decstyle] table [x=s, y=eta_dec, col sep=comma] {content/plots/01_gamma_plot.csv};
      \addlegendentry{$\|\mathcal{D}\mathcal{T}_{\mathrm{dec}}(K^\star(s))\|_{\pi,w,\infty}$}

      \addplot[seqBstyle] table [x=s, y=eta_seqB, col sep=comma] {content/plots/01_gamma_plot.csv};
      \addlegendentry{$\|\mathcal{D}\mathcal{T}_{\mathrm{seq}}^{\pi^{(2)}}(K^\star(s))\|_{\pi^{(2)},w,\infty}$}

      \addplot[seqAstyle] table [x=s, y=eta_seqA, col sep=comma] {content/plots/01_gamma_plot.csv};
      \addlegendentry{$\|\mathcal{D}\mathcal{T}_{\mathrm{seq}}^{\pi^{(1)}}(K^\star(s))\|_{\pi^{(1)},w,\infty}$}
    
      \addplot[thresholdstyle] coordinates {(0,1) (2.6,1)};
    \end{axis}
    \end{tikzpicture}
    \caption{Weighted block-\(\ell_\infty\) comparison quantities associated with
    Theorem~\ref{thm:seq-no-worse} along the diffusion sweep.}
    \label{fig:gamma_plot}
\end{figure}
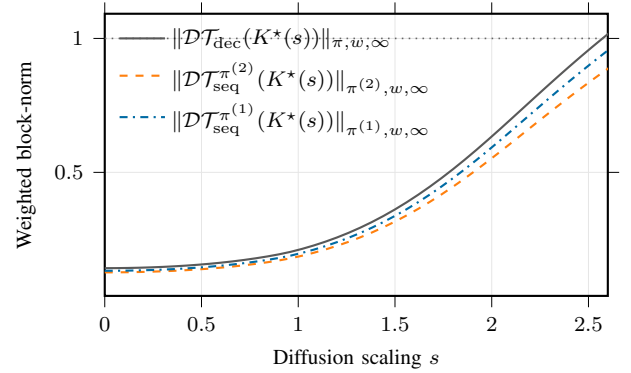

\section{Conclusion}
\label{sec:conclusion}

This paper has proposed a sequential policy-iteration method for
infinite-horizon stochastic linear--quadratic differential games with
state- and control-dependent noise. The player-wise sequential
construction preserves mean-square stability at every intermediate update,
yielding a well-posed iteration on the set of mean-square-stabilizing gain profiles.
Its fixed points are precisely the stabilizing feedback Nash equilibria, so that every run converging within this set certifies the solvability of the coupled equilibrium equations for the instance at hand.
From the Fr\'echet derivative of the iteration map, we established a
sufficient criterion for local convergence.
A weighted block-norm comparison further delimits the equilibria at which sequential
updating is guaranteed to contract at least as strongly as updating all
players simultaneously from a common gain profile. The same derivative shows
that the order in which the players are updated can decide whether an
equilibrium is locally attracting at all.
To address the initialization requirement of infinite-horizon policy
iteration, we further developed a homotopy initialization that preserves
mean-square stability along the continuation path and provides a
stabilizing initial gain profile for the full \(N\)-player iteration.
The numerical example confirms the analysis for a three-player game.
Increasing diffusion weakens local contraction, changes which update order is
locally preferred, and causes contraction to be lost at different diffusion
levels for different orders.

\appendices

\section{Explicit Fr\'echet blocks}
\label{app:frechet-blocks}

This appendix records explicit expressions for the Fr\'echet blocks
appearing in Lemma~\ref{lem:reduced-playerwise}. 
For \(X\in\mathbb S^n\), \(H_i\in\mathbb R^{m_i\times n}\), and
\(\Delta K=(\Delta K_j)_{j\in\mathcal I}\in\mathcal K\), the partial
Fr\'echet derivatives of \(\mathcal{E}_i\) and \(\mathcal{G}_i\), evaluated at
\((P_i(K),\mathcal U_i(K),K)\), are given as follows.

\medskip
\noindent\textup{(i) Evaluation residual.}
\begin{align}
&\mathcal{D}_P\mathcal{E}_i(K)[X]
=
A_{\mathrm{cl}}(K)^\top X
+
X A_{\mathrm{cl}}(K)
+
C_{\mathrm{cl}}(K)^\top X\,C_{\mathrm{cl}}(K),
\label{eq:app-EPi}
\\
&\mathcal{D}_K\mathcal{E}_i(K)[\Delta K]
=
-\sum_{j\in\mathcal I}
\left(
\Delta K_j^\top B_j^\top P_i
+
P_i B_j \Delta K_j
\right)
\nonumber\\
&\quad \qquad 
-\sum_{j\in\mathcal I}
\left(
\Delta K_j^\top D_j^\top P_i C_{\mathrm{cl}}(K)
+
C_{\mathrm{cl}}(K)^\top P_i D_j \Delta K_j
\right)
\nonumber\\
&\quad \qquad 
+\sum_{j\in\mathcal I}
\left(
K_j^\top R_{ij}\Delta K_j
+
\Delta K_j^\top R_{ij}K_j
\right).
\label{eq:app-EKi}
\end{align}

\medskip
\noindent\textup{(ii) Fixed-value stationarity residual.}
\begin{align}
\mathcal{D}_P\mathcal{G}_i(K)[X]
&=
D_i^\top X D_i\,\widehat K_i + \sum_{j\in\mathcal I\setminus\{i\}} D_i^\top X D_j K_j
 \notag \\
&\quad
-
B_i^\top X
-
D_i^\top X C,
\label{eq:app-GPi}
\\
\mathcal{D}_{\widehat K_i}\mathcal{G}_i(K)[H_i]
&=
\bigl(R_{ii}+D_i^\top P_i D_i\bigr)H_i,
\label{eq:app-GKhatii}
\\
\mathcal{D}_K\mathcal{G}_i(K)[\Delta K]
&=
\sum_{j\in\mathcal I\setminus\{i\}} D_i^\top P_i D_j \Delta K_j.
\label{eq:app-GKi}
\end{align}

\section{Proof of Corollary~\ref{cor:Jinv-block-comparison}}
\label{app:proof-explicit-contraction}
\begin{proof}
By \eqref{eq:Mpi-definition} and \eqref{eq:transported-block-norm},
$
\|\mathcal{D}\mathcal{T}_{\mathrm{inv}}\|_{\pi,w,\infty}
=
\|\mathcal{D}\mathcal{T}_{\mathrm{dec}}^{\pi}\|_{w,\infty}.$
Moreover, the \((i,j)\)-block of \(\mathcal{D}\mathcal{T}_{\mathrm{dec}}^{\pi}\) is
$
(\mathcal{D}\mathcal{T}_{\mathrm{dec}}^{\pi})_{ij}
=
[\mathcal{D}\mathcal{T}_{\mathrm{inv}}]_{\pi(i)\pi(j)}$, for
$i,j\in\mathcal I.
$
Hence, by \eqref{eq:gammaii}, \eqref{eq:gammaij},
$
\|(\mathcal{D}\mathcal{T}_{\mathrm{dec}}^{\pi})_{ij}\|
\le
\gamma_{\pi(i)\pi(j)}$, for
$ i,j\in\mathcal I.
$
Let \(\Delta K^{\pi}\in\mathcal K^\pi\). For each \(i\in\mathcal I\),
\begin{align}
\frac{\|(\mathcal{D}\mathcal{T}_{\mathrm{dec}}^{\pi} \Delta K^{\pi})_i\|}{w_i}
&=
\frac{\|\sum_{j\in\mathcal I} (\mathcal{D}\mathcal{T}_{\mathrm{dec}}^{\pi})_{ij}\Delta K^{\pi}_j\|}{w_i}
\notag\\
&\le
\sum_{j\in\mathcal I}
\frac{\|(\mathcal{D}\mathcal{T}_{\mathrm{dec}}^{\pi})_{ij}\|\,\|\Delta K^{\pi}_j\|}{w_i}
\notag\\
&\le
\sum_{j\in\mathcal I}
\frac{w_j}{w_i}\,\gamma_{\pi(i)\pi(j)}\,
\frac{\|\Delta K^{\pi}_j\|}{w_j}
\notag\\
&\le
\left(
\sum_{j\in\mathcal I} \frac{w_j}{w_i}\,\gamma_{\pi(i)\pi(j)}
\right)
\|\Delta K^{\pi}\|_{w,\infty}.
\label{eq:proof-row-estimate}
\end{align}
By \eqref{eq:Jinv-weighted-comparison},
$
\frac{\|(\mathcal{D}\mathcal{T}_{\mathrm{dec}}^{\pi} \Delta K^{\pi})_i\|}{w_i}
\le
q\,\|\Delta K^{\pi}\|_{w,\infty}$,
for $ i\in\mathcal I.
$
Taking the maximum over \(i\) yields
$
\|\mathcal{D}\mathcal{T}_{\mathrm{dec}}^{\pi} \Delta K^{\pi}\|_{w,\infty}
\le
q\,\|\Delta K^{\pi}\|_{w,\infty}$, with $ \Delta K^{\pi}\in\mathcal K^\pi.
$
Therefore,
\begin{equation}
\label{eq:proof-Jinv-contraction}
\|\mathcal{D}\mathcal{T}_{\mathrm{inv}}\|_{\pi,w,\infty}
=
\|\mathcal{D}\mathcal{T}_{\mathrm{dec}}^{\pi}\|_{w,\infty}
\le q<1.
\end{equation}
Since \(\mathcal{D}\mathcal{T}_{\mathrm{inv}}=\mathcal{D}\mathcal{T}_{\mathrm{dec}}(K^\star)\) by Proposition~\ref{prop:decoupled-jacobian}, Theorem~\ref{thm:seq-no-worse} implies
\begin{equation}
\label{eq:proof-seq-contraction}
\|\mathcal{D}\mathcal{T}_{\mathrm{seq}}^{\pi}(K^\star)\|_{\pi,w,\infty}
\le
\|\mathcal{D}\mathcal{T}_{\mathrm{inv}}\|_{\pi,w,\infty}
\le q<1.
\end{equation}
The spectral radius bound $\rho(\mathcal{D}\mathcal{T}_{\mathrm{seq}}^{\pi}(K^\star))<1$ follows since $\rho\le\|\cdot\|_{\pi,w,\infty}$.
\end{proof}

\section*{References}
\bibliographystyle{IEEEtran}
\bibliography{bibliography/refs}

@book{basar_olsder_1998,
  author    = {Ba{\c{s}}ar, Tamer and Olsder, Geert Jan},
  title     = {Dynamic Noncooperative Game Theory},
  edition   = {2nd},
  series    = {Classics in Applied Mathematics},
  volume    = {23},
  publisher = {SIAM},
  address   = {Philadelphia, PA, USA},
  year      = {1998},
  doi       = {10.1137/1.9781611971132}
}

@book{friedman_2013,
  author    = {Friedman, Avner},
  title     = {Differential Games},
  publisher = {Dover Publications},
  address   = {Mineola, NY, USA},
  year      = {2013}
}

@book{yuksel_basar_2013,
  author    = {Y{\"u}ksel, Serdar and Ba{\c{s}}ar, Tamer},
  title     = {Stochastic Networked Control Systems: Stabilization and Optimization under Information Constraints},
  publisher = {Springer},
  year      = {2013},
  doi       = {10.1007/978-1-4614-7085-4}
}

@article{lewis_vrabie_2009_rl_adp,
  author  = {Lewis, Frank L. and Vrabie, Draguna},
  title   = {Reinforcement Learning and Adaptive Dynamic Programming for Feedback Control},
  journal = {{IEEE} Circuits and Systems Magazine},
  volume  = {9},
  number  = {3},
  pages   = {32--50},
  year    = {2009},
  doi     = {10.1109/MCAS.2009.933854}
}

@article{vamvoudakis_lewis_2011_online_hjb_games,
  author  = {Vamvoudakis, Kyriakos G. and Lewis, Frank L.},
  title   = {Multi-player non-zero-sum games: Online adaptive learning solution of coupled {Hamilton--Jacobi} equations},
  journal = {Automatica},
  volume  = {47},
  number  = {8},
  pages   = {1556--1569},
  year    = {2011},
  doi     = {10.1016/j.automatica.2011.03.012}
}

@article{hu_lauriere_2024_survey_ml_stochastic_games,
  author  = {Hu, Ruimeng and Lauri{\`e}re, Mathieu},
  title   = {Recent developments in machine learning methods for stochastic control and games},
  journal = {Numerical Algebra, Control and Optimization},
  volume  = {14},
  number  = {3},
  pages   = {435--525},
  year    = {2024},
  doi     = {10.3934/naco.2024031}
}

@article{chen_lewis_li_2022_homotopic_pi,
  author  = {Chen, Ci and Lewis, Frank L. and Li, Bo},
  title   = {Homotopic policy iteration-based learning design for unknown linear continuous-time systems},
  journal = {Automatica},
  volume  = {138},
  pages   = {110153},
  year    = {2022},
  doi     = {10.1016/j.automatica.2021.110153}
}

@article{moon_wang_basar_2026_survey_lq_sdg,
  author  = {Moon, Jun and Wang, Bing-Chang and Ba{\c{s}}ar, Tamer},
  title   = {A selective survey of recent results on linear-quadratic stochastic differential games},
  journal = {International Journal of Control, Automation, and Systems},
  volume  = {24},
  pages   = {1--29},
  year    = {2026},
  doi     = {10.1007/s12555-026-00003-y},
  number  = {1}
}

@article{nortmann_monti_sassano_mylvaganam2024_tac_iterative_datadriven_lq_games,
  author  = {Nortmann, Benita and Monti, Andrea and Sassano, Mario and Mylvaganam, Thulasi},
  title   = {Nash equilibria for linear quadratic discrete-time dynamic games via iterative and data-driven algorithms},
  journal = {{IEEE} Transactions on Automatic Control},
  volume  = {69},
  number  = {10},
  pages   = {6561--6575},
  year    = {2024},
  doi     = {10.1109/TAC.2024.3375249}
}

@article{saito_takahashi_2019_automatica,
  author  = {Saito, Takashi and Takahashi, Akihiko},
  title   = {Stochastic differential game in high frequency market},
  journal = {Automatica},
  volume  = {104},
  pages   = {111--125},
  year    = {2019},
  doi     = {10.1016/j.automatica.2019.02.051}
}

@article{liu_huang_2025_tac_aggregative,
  author  = {Liu, Zhixin and Huang, Jie},
  title   = {Distributed Nash Equilibrium Seeking in Aggregative Games Over Jointly Connected and Weight-Balanced Networks},
  journal = {{IEEE} Transactions on Automatic Control},
  volume  = {70},
  number  = {5},
  pages   = {3486--3493},
  year    = {2025},
  doi     = {10.1109/TAC.2024.3520809}
}

@inproceedings{varga_inga_lemmer_hohmann_2021_ccta,
  author    = {Varga, Balint and Inga, Jairo and Lemmer, Marius and Hohmann, S{\"o}ren},
  title     = {Ordinal Potential Differential Games to Model Human-Machine Interaction in Vehicle-Manipulators},
  booktitle = {Proc.\ {IEEE} Conf.\ on Control Technology and Applications ({CCTA})},
  year      = {2021},
  pages     = {728--734},
  doi       = {10.1109/CCTA48906.2021.9658788}
}

@inproceedings{kille_leibold_karg_varga_hohmann_2024_roman,
  author    = {Kille, Sven and Leibold, Marion and Karg, Philipp and Varga, Balint and Hohmann, S{\"o}ren},
  title     = {Human-Variability-Respecting Optimal Control for Physical Human-Machine Interaction},
  booktitle = {Proc.\ 33rd {IEEE} Int.\ Conf.\ on Robot and Human Interactive Communication ({RO-MAN})},
  year      = {2024},
  pages     = {1595--1602},
  doi       = {10.1109/RO-MAN60168.2024.10731297}
}

@article{todorov_2005_neural_computation,
  author  = {Todorov, Emanuel},
  title   = {Stochastic Optimal Control and Estimation Methods Adapted to the Noise Characteristics of the Sensorimotor System},
  journal = {Neural Computation},
  volume  = {17},
  number  = {5},
  pages   = {1084--1108},
  year    = {2005},
  doi     = {10.1162/0899766053491887}
}

@book{yong_zhou_1999,
  author    = {Yong, Jiongmin and Zhou, Xun Yu},
  title     = {Stochastic Controls: {Hamiltonian} Systems and {HJB} Equations},
  publisher = {Springer-Verlag},
  address   = {New York, NY, USA},
  year      = {1999},
  doi       = {10.1007/978-1-4612-1466-3}
}

@article{starr_ho_1969,
  author  = {Starr, A. W. and Ho, Y.-C.},
  title   = {Nonzero-sum differential games},
  journal = {Journal of Optimization Theory and Applications},
  volume  = {3},
  number  = {3},
  pages   = {184--206},
  year    = {1969},
  doi     = {10.1007/BF00929443}
}

@article{kleinman_1969,
  author  = {Kleinman, D. L.},
  title   = {Optimal stationary control of linear systems with control-dependent noise},
  journal = {{IEEE} Transactions on Automatic Control},
  volume  = {14},
  number  = {6},
  pages   = {673--677},
  year    = {1969},
  doi     = {10.1109/TAC.1969.1099303}
}

@article{mclane_1971,
  author  = {McLane, Peter},
  title   = {Optimal stochastic control of linear systems with state- and control-dependent disturbances},
  journal = {{IEEE} Transactions on Automatic Control},
  volume  = {16},
  number  = {6},
  pages   = {793--798},
  year    = {1971},
  doi     = {10.1109/TAC.1971.1099828}
}

@article{willems_willems_1976,
  author  = {Willems, Jan L. and Willems, Jan C.},
  title   = {Feedback stabilizability for stochastic systems with state and control dependent noise},
  journal = {Automatica},
  volume  = {12},
  number  = {3},
  pages   = {277--283},
  year    = {1976},
  doi     = {10.1016/0005-1098(76)90029-7}
}

@article{freiling_1996,
  author  = {Freiling, Gerhard and Jank, Gerhard and Abou-Kandil, Hisham},
  title   = {On global existence of solutions to coupled matrix {Riccati} equations in closed-loop {Nash} games},
  journal = {{IEEE} Transactions on Automatic Control},
  volume  = {41},
  number  = {2},
  pages   = {264--269},
  year    = {1996},
  doi     = {10.1109/9.481532}
}

@incollection{li_gajic_1995,
  author    = {Li, Tung-Yi and Gaji{\'c}, Zoran},
  title     = {Lyapunov iterations for solving coupled algebraic {Riccati} equations of {Nash} differential games and algebraic {Riccati} equations of zero-sum games},
  booktitle = {New Trends in Dynamic Games and Applications},
  editor    = {Olsder, G. J.},
  publisher = {Birkh{\"a}user},
  year      = {1995},
  pages     = {333--351},
  doi     = {10.1007/978-1-4612-4274-1_17}
}

@article{puterman_brumelle_1979_pi_newton,
  author  = {Puterman, Martin L. and Brumelle, Shelby L.},
  title   = {On the Convergence of Policy Iteration in Stationary Dynamic Programming},
  journal = {Mathematics of Operations Research},
  volume  = {4},
  number  = {1},
  pages   = {60--69},
  year    = {1979},
  doi     = {10.1287/moor.4.1.60}
}

@inproceedings{guan_salizzoni_kamgarpour_summers_2024_pi_lqdg,
  author    = {Guan, Yutao and Salizzoni, Giorgia and Kamgarpour, Maryam and Summers, Tyler H.},
  title     = {A Policy Iteration Algorithm for {N}-player General-Sum Linear Quadratic Dynamic Games},
  booktitle = {Proc.\ 63rd {IEEE} Conf.\ on Decision and Control ({CDC})},
  pages     = {1725--1730},
  year      = {2024},
  doi       = {10.1109/CDC56724.2024.10886048}
}

@article{chen_chen_lewis_xie_mcpi_nash_dg,
  author  = {Chen, Yuzhe and Chen, Ci and Lewis, Frank L. and Xie, Shengli},
  title   = {Multiplayer Cascaded Policy Iteration for {Nash} Differential Games},
  journal = {{IEEE} Transactions on Automatic Control},
  volume  = {71},
  number  = {2},
  pages   = {1038--1053},
  year    = {2026},
  doi     = {10.1109/TAC.2025.3605273}
}

@article{chen_chen_lewis_2026_automatica_mcpi_nonlinear,
  author  = {Chen, Yuzhe and Chen, Ci and Lewis, Frank L.},
  title   = {Non-zero-sum games in continuous-time nonlinear systems: Multiplayer cascaded solutions},
  journal = {Automatica},
  volume  = {186},
  year    = {2026},
  pages   = {112828},
  doi     = {10.1016/j.automatica.2026.112828}
}

@book{nisio2015,
  author    = {Nisio, Makiko},
  title     = {Stochastic Control Theory: Dynamic Programming Principle},
  edition   = {2nd},
  series    = {Probability Theory and Stochastic Modelling},
  volume    = {72},
  publisher = {Springer},
  address   = {Tokyo, Japan},
  year      = {2015},
  doi     = {10.1007/978-4-431-55123-2}
}

@book{oksendal_2003,
  author    = {{\O}ksendal, Bernt},
  title     = {Stochastic Differential Equations: An Introduction with Applications},
  edition   = {6th},
  publisher = {Springer},
  address   = {Berlin, Germany},
  year      = {2003},
  doi       = {10.1007/978-3-642-14394-6}
}

@inproceedings{sagara_mukaidani_yamamoto2007_state_dep_noise,
  author    = {Sagara, Muneomi and Mukaidani, Hiroaki and Yamamoto, Toru},
  title     = {Stochastic {Nash} Games with State-Dependent Noise},
  booktitle = {Proc.\ {SICE} Annual Conference},
  address   = {Kagawa University, Japan},
  pages     = {1635--1638},
  year      = {2007},
  doi       = {10.1109/SICE.2007.4421245}
}

@article{ivanov_tanov_2018_iterative_lq_stochastic_games,
  author  = {Ivanov, Ivelin G. and Tanov, Vladislav K.},
  title   = {An iterative method for an equilibrium point of linear quadratic stochastic differential games with state and control-dependent noise},
  journal = {Annals of the Academy of Romanian Scientists, Series on Mathematics and its Applications},
  volume  = {10},
  number  = {2},
  pages   = {202--210},
  year    = {2018}
}

@article{han_hu_2020_deep_fp,
  author  = {Han, Jiequn and Hu, Ruimeng},
  title   = {Deep Fictitious Play for Finding {Markovian} {Nash} Equilibrium in Multi-Agent Games},
  journal = {Proceedings of Machine Learning Research},
  volume  = {107},
  pages   = {221--245},
  year    = {2020}
}

@article{han_hu_long_2020_convergence_dfp,
  author  = {Han, Jiequn and Hu, Ruimeng and Long, Jihao},
  title   = {Convergence of deep fictitious play for stochastic differential games},
  journal = {Frontiers of Mathematical Finance},
  volume  = {1},
  number  = {2},
  pages   = {287--319},
  year    = {2022},
  doi     = {10.3934/fmf.2021011}
}

@book{damm2004,
  author    = {Damm, Tobias},
  title     = {Rational Matrix Equations in Stochastic Control},
  series    = {Lecture Notes in Control and Information Sciences},
  volume    = {297},
  publisher = {Springer-Verlag},
  address   = {Berlin, Germany},
  year      = {2004}
}

@book{OrtegaRheinboldt,
  author    = {Ortega, James M. and Rheinboldt, Werner C.},
  title     = {Iterative Solution of Nonlinear Equations in Several Variables},
  publisher = {Academic Press},
  address   = {New York, USA},
  year      = {1970}
}

@book{Deuflhard,
  author    = {Deuflhard, Peter},
  title     = {Newton Methods for Nonlinear Problems: Affine Invariance and Adaptive Algorithms},
  series    = {Springer Series in Computational Mathematics},
  publisher = {Springer-Verlag},
  address   = {Berlin, Germany},
  year      = {2011},
  doi     = {10.1007/978-3-642-23899-4}
}

@book{HornJohnson,
  author    = {Horn, Roger A. and Johnson, Charles R.},
  title     = {Matrix Analysis},
  edition   = {2nd},
  publisher = {Cambridge University Press},
  address   = {Cambridge, UK},
  year      = {2012},
  doi       = {10.1017/CBO9781139020411}
}

@book{engwerda2005,
  author    = {Engwerda, Jacob},
  title     = {{LQ} Dynamic Optimization and Differential Games},
  publisher = {John Wiley \& Sons, Ltd},
  address   = {Chichester, England},
  year      = {2005}
}

@article{zhu_zhang2013_state_control_noise,
  author  = {Zhu, Huainian and Zhang, Chengke},
  title   = {Infinite time horizon nonzero-sum linear quadratic stochastic differential games with state and control-dependent noise},
  journal = {Journal of Control Theory and Applications},
  volume  = {11},
  number  = {4},
  pages   = {629--633},
  year    = {2013},
  doi     = {10.1007/s11768-013-1182-3}
}

@article{feingold_varga_1962,
  author  = {Feingold, David G. and Varga, Richard S.},
  title   = {Block diagonally dominant matrices and generalizations of the {Gerschgorin} circle theorem},
  journal = {Pacific Journal of Mathematics},
  volume  = {12},
  number  = {4},
  pages   = {1241--1250},
  year    = {1962},
  doi     = {10.2140/pjm.1962.12.1241}
}

@book{varga_2000,
  author    = {Varga, Richard S.},
  title     = {Matrix Iterative Analysis},
  edition   = {2nd rev.\ and expanded},
  publisher = {Springer},
  address   = {Berlin, Germany},
  year      = {2000},
  doi       = {10.1007/978-3-642-05156-2}
}

@article{householder_1956,
  author  = {Householder, Alston S.},
  title   = {On the Convergence of Matrix Iterations},
  journal = {Journal of the {ACM}},
  volume  = {3},
  number  = {4},
  pages   = {314--324},
  year    = {1956},
  doi     = {10.1145/320843.320851}
}

@misc{thoemmes2026_ifac_policy_gradient,
  author = {Th{\"o}mmes, Felix and G{\"u}nther, Lucas and Handwerker, Karl and Kr{\"u}ger, Philipp and Varga, Balint and Hohmann, S{\"o}ren},
  title  = {Policy Gradient Methods for Continuous-Time Linear Quadratic Games},
  year   = {2026},
  note   = {{IFAC}-PapersOnLine, in press}
}

@misc{Lucas,
  author = {G{\"u}nther, Lucas and Handwerker, Karl and Th{\"o}mmes, Felix and Varga, Balint and Hohmann, S{\"o}ren},
  title  = {Infinite-Horizon Inverse Linear-Quadratic Differential Games with State- and Control-Dependent Noise},
  year   = {2026},
  note   = {arXiv preprint}
}

@article{mukaidani2009_automatica_soft_constrained,
  author  = {Mukaidani, Hiroaki},
  title   = {Soft-constrained stochastic {Nash} games for weakly coupled large-scale systems},
  journal = {Automatica},
  volume  = {45},
  number  = {5},
  pages   = {1272--1279},
  year    = {2009},
  doi     = {10.1016/j.automatica.2008.12.020}
}

@article{sun_jiang_zhang2012_discrete_lq_games,
  author  = {Sun, Huiying and Jiang, Liuyang and Zhang, Weihai},
  title   = {Infinite horizon linear quadratic differential games for discrete-time stochastic systems},
  journal = {Journal of Control Theory and Applications},
  volume  = {10},
  number  = {3},
  pages   = {391--396},
  year    = {2012},
  doi     = {10.1007/s11768-012-1004-z}
}

@article{hambly_xu_yang2022_policy_gradient_lq_games,
  author  = {Hambly, Ben and Xu, Renyuan and Yang, Huining},
  title   = {Policy Gradient Methods Find the {Nash} Equilibrium in {N}-player General-sum Linear-quadratic Games},
  journal = {Journal of Machine Learning Research},
  volume  = {24},
  number  = {139},
  pages   = {1--56},
  year    = {2023}
}

@misc{andersson_andersson_ljung2026_fp_fbsde,
  author = {Andersson, Adam and Andersson, Kristoffer and Ljung, Per},
  title  = {Convergence of fictitious play for fully coupled {FBSDEs} in finite-player stochastic differential games},
  year   = {2026},
  note   = {arXiv preprint arXiv:2607.08861},
  doi    = {10.48550/arXiv.2607.08861}
}

\end{document}